\documentclass[a4paper,numberwithinsect,USenglish]{lipics-v2021}

\title{On the Complexity of Fundamental Problems for DAG-Compressed Graphs}

\author{Florian Chudigiewitsch}{Institute for Theoretical Computer Science, Universität zu Lübeck, Germany}{fch@tcs.uni-luebeck.de}{https://orcid.org/0000-0003-3237-1650}{}
\author{Till Tantau}{Institute for Theoretical Computer Science, Universität zu Lübeck, Germany}{tantau@tcs.uni-luebeck.de}{}{}
\author{Felix Winkler}{Universität zu Lübeck, Germany}{}{}{}

\authorrunning{F. Chudigiewitsch and T. Tantau and F. Winkler}
\Copyright{F. Chudigiewitsch and T. Tantau and F. Winkler} 


\category{} 

\relatedversion{} 

\nolinenumbers 

\hideLIPIcs  

\EventEditors{John Q. Open and Joan R. Access} \EventNoEds{2}
\EventLongTitle{42nd Conference on Very Important Topics (CVIT 2016)}
\EventShortTitle{CVIT 2016} \EventAcronym{CVIT} \EventYear{2016}
\EventDate{December 24--27, 2016} \EventLocation{Little Whinging, United
Kingdom} \EventLogo{} \SeriesVolume{42} \ArticleNo{23}

\theoremstyle{plain}
\newtheorem{problem}[theorem]{Problem}

\newtheorem{invariant}[theorem]{Invariant}

\usepackage{listings}
\lstdefinelanguage{pseudocode}{ morekeywords={ algorithm,method,new,and,not,
  if,then,else,while,do,repeat,until,seq, seqdo,return,call,
  for,pardo,foreach,print,output,input,exit,
  break,loop,end,begin,goto,par,global,local,
  read,write,stop,idle,procedure,function, throw,catch }, sensitive=true,
  morecomment=[l]{//}, morestring=[b]", morestring=[s]{``}{''}, }
  \lstdefinestyle{pseudocode}{ language=pseudocode, basicstyle=\small\rmfamily,
  commentstyle=\upshape\color{black!50}, keywordstyle=\bfseries\itshape,
  identifierstyle=\itshape, stringstyle=\rmfamily, columns=fullflexible,
  mathescape, literate={<-}{{$\gets$\ }}2, numbers=left,
  numberstyle=\scriptsize\sffamily, }

\usepackage{booktabs}

\usepackage{mathtools}

\newcommand\Class[1]{%
  \mathchoice%
  {\text{\normalfont\small$\mathrm{#1}$}}%
  {\text{\normalfont\small$\mathrm{#1}$}}%
  {\text{\normalfont$\mathrm{#1}$}}%
  {\text{\normalfont$\mathrm{#1}$}}%
}

{%
  \leavevmode\nobreak\par
  \begin{list}%
    {}%
    {%
      \def\labelstyle{\itshape} \setlength{\topsep}{0pt}%
      \settowidth{\labelwidth}{\labelstyle Parameter:}%
      \setlength{\leftmargin}{\labelwidth}%
      \addtolength{\leftmargin}{\labelsep}%
      \setlength{\itemsep}{0pt}%
      \setlength{\parsep}{0pt}%
    }%
    }%
    {%
  \end{list}%
}

\newcommand{\Lang}[1]{\text{\normalfont\textsc{#1}}}

\usepackage{bm}
\newcommand{\ori}[1]{\bm{#1}}
\newcommand{\abbrv}[1]{{\smaller{#1}}}

\makeatletter
\def\smaller{\fontsize{\dimexpr\f@size pt-1pt\end}{\f@size}\selectfont}
\makeatother

\newcommand\coloneq{\mathrel{\raise.4pt\hbox{:}{=}}}
\newcommand\eqcolon{\mathrel{{=}\raise.4pt\hbox{:}}}

\usepackage[move]{tcs-automove}

\usepackage{xfrac}

\makeatletter
\newtheorem*{rep@theorem}{\rep@title}
\newcommand{\newreptheorem}[2]{%
\newenvironment{rep#1}[1]{%
 \def\rep@title{#2 \ref{##1}}%
 \begin{rep@theorem}}%
 {\end{rep@theorem}}}
\makeatother
\newreptheorem{theorem}{Theorem}

\usepackage{tikz}
\usetikzlibrary{graphs, arrows.meta, quotes, positioning, backgrounds, shapes, decorations.pathmorphing,shapes.multipart,fadings}
\tikzset{ every picture/.style={semithick},%
  > = {Stealth[round,sep]}, grayed/.style = { draw=black!30 },
  node/.style={ draw, circle, minimum size=6mm, inner sep=0.5pt,
    font=\footnotesize, fill=white },%
  vertex/.style={ 
    node },
  small node/.style={ node,minimum size=4.5pt, inner sep=0pt, outer sep=0pt,
    font=\tiny },%
  on/.style={fill=white,inner sep=-.4pt,circle},
  mapsto/.style={|[sep]->,blue!50!black},
  arbitrary/.style={dashed,draw=black!50}, present/.style={black},
  missing/.style={black!10}, virtual node/.style={node,draw=none,fill=none},
  white node/.style={node}, black node/.style={node, fill=black, text=white},
  gray node/.style={node, arbitrary, fill=black!20}, red node/.style={node,
  red!60!black, text=white}, blue node/.style={node, blue!70!black, text=white},
  green node/.style={node, green!50!black, text=white},
  hilight/.style={red!80!black}, new/.style={ line width=.4pt, double
  distance=.8pt, draw=white, double=lipicsYellow, }, new edge/.style={
  draw=lipicsYellow, thick },
  sinkVertex/.style={node, font=\bfseries},
  cluster/.style={arc}}

\definecolor{darkgreen}{RGB}{0,150,0}
\definecolor{darkred}{RGB}{210,0,0}
\definecolor{darkblue}{RGB}{0,0,210}

  \newcommand{\doublecross}[1]{
  \draw[darkred, thick] ({#1}) -- ++(-0.36, 0.36) -- ++(0.72, -0.72);
  \draw[darkred,thick] ({#1}) -- ++(-0.36, -0.36) -- ++(0.72, 0.72);
  }

\tikzfading[name=fade left,
  left color=transparent!0, right color=transparent!100]
\tikzfading[name=fade right,
  left color=transparent!100, right color=transparent!0]
\tikzset{
  sepline/.style={thick,black!25},
  compressedUn/.style={double,double equal sign distance, thick},
  compressed/.style={double,double equal sign distance,->, thick},
  compressedSelfLoop/.style={double, double equal sign distance, thick, loop above },
  dedge/.style={->,thick},
  dedgeleft/.style={dedge, bend left= 10},
  arc/.style={->,black!50},
  clusterVertex/.style={node,draw=black!50},
  pickedOne/.style={draw=green!70!black},
}

\definecolor{darkred}{RGB}{139,0,0}
\definecolor{darkblue}{RGB}{59,86,255}
\definecolor{lightred}{RGB}{255,230,230}
\definecolor{lightblue}{RGB}{230,230,246}

\usepackage{mathtools}

\newcommand{\TSnice}{\widehat{T}'}

\newcommand{\Algo}[1]{{\text{\upshape\scshape#1}}}

\begin{document}

\maketitle

\begin{abstract}
  A \emph{\abbrv{DAG} compression} of a (typically dense) graph is a
  simple data structure that stores how vertex clusters are connected,
  where the clusters are described indirectly as sets of reachable
  sinks in a directed acyclic
  graph (\abbrv{DAG}). They generalize tree compressions, where the
  clusters form a tree-like hierarchy, and we give the first proof
  that \abbrv{DAG} compressions can achieve better compressions than
  tree compressions. Our interest in \abbrv{DAG} compression stems
  from the fact that several simple standard algorithms, like
  breadth-first search on graphs, can be implemented so that they 
  work directly on the compressed rather than on the original
  graph and so that, crucially, the runtime is relative to the
  (typically small)
  size of the compressed graph. We add another entry to the list
  of algorithms where this is possible, by showing that Kruskal's
  algorithm for computing minimum spanning trees can be adapted to work
  directly on \abbrv{DAG} compressions. On the negative side, we
  answer the central open problem from previous work, namely how hard
  it is to compute a minimum-size \abbrv{DAG} compression for a given
  graph: This is $\Class{NP}$-hard; and this is even the case
  for the dynamic setting, where we must update the \abbrv{DAG} compression optimally when
  a single edge is added or deleted in the input graph.
\end{abstract}

\newpage{}

\section{Introduction}

Data compression is an indispensable tool for
processing and storing huge amounts of data and has become a major research topic in
theoretical computer science~\cite{BlandfordBK03, BoldiV04, BouritsasLKB21,
ChierichettiKLMPR09, DhulipalaKKOPS16, ShunDB15, Versari21}.
We are interested in compressing \emph{dense graphs,} which in modern applications
can easily encompass billions of edges, in such a way that we can
run fundamental algorithms \emph{directly on the compressed graphs}
without needing to decompress them.
A particularly simple way of compressing directed graphs $\ori G =
(\ori V, \ori E)$ was recently introduced~\cite{BannachMT24} in the
form of \emph{\abbrv{DAG} compressions.} They are triples $D = (V,A,E)$
such that (we use boldface for the original graph and reserve the
standard font for the compression), firstly, $(V,A)$ is a 
directed acyclic graph (\abbrv{DAG}) whose sinks are exactly the
vertices in $\ori V \subseteq 
V$. Each vertex $v \in V$ represents a \emph{cluster $\ori C_D(v)
\subseteq \ori V$,} defined as the set of sinks reachable from~$v$
in~$(V,A)$ (we omit the subscript when $D$ is clear from
context). Secondly, $E \subseteq  V \times V$ is a \emph{set of 
compression edges} so that $\ori E = \bigcup_{(u, v)  \in E} \ori C(u)
\times \ori C(v)$, meaning that each compression edge  
$(u, v)$ encodes the presence  of all possible edges from vertices
in~$\ori C(u)$ to vertices in~$\ori C(v)$ in~$\ori G$. 
Figure~\ref{example:DAGCompression} depicts an example.

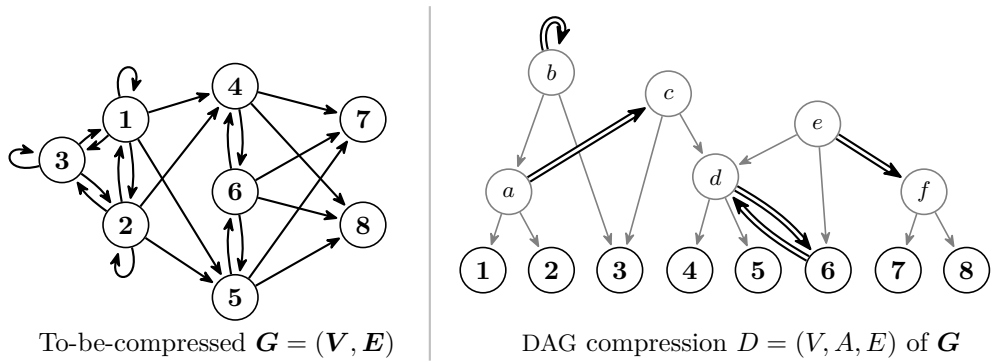
\begin{figure}[th]
  \centering

  \begin{tikzpicture}[node distance=1.4cm]
    \begin{scope}[xshift=4.2cm]
      \tikzset{node distance=0.9cm, fill=white}
      
      \node[sinkVertex] (1) {$\ori{1}$};
      \node[sinkVertex, right of=1] (2) {$\ori{2}$};
      \node[sinkVertex, right of=2] (3) {$\ori{3}$};
      \node[sinkVertex, right = 0.3cm of 3] (4) {$\ori{4}$};
      \node[sinkVertex, right of=4] (5) {$\ori{5}$};
      \node[sinkVertex, right= 0.3cm of 5] (6) {$\ori{6}$};
      \node[sinkVertex, right= 0.3cm of 6] (7) {$\ori{7}$};
      \node[sinkVertex, right of=7] (8) {$\ori{8}$};

      \node[clusterVertex, above right =0.6cm and -0.1cm of 1] (a) {$a$};
      \node[clusterVertex, above =2cm of 2] (b) {$b$};
      \node[clusterVertex, above right =1.9cm and 0.2cm of 3] (c) {$c$};
      \node[clusterVertex, above right =0.8cm and -0.1cm of 4] (d) {$d$};
      \node[clusterVertex, above left  = 1.5cm and -0.3cm of 6] (e) {$e$};
      \node[clusterVertex,  above right = 0.6cm and -0.1cm of 7] (f) {$f$};

      \draw (a) edge[arc] (1);
      \draw (a) edge[arc] (2);
      \draw (b) edge[arc] (a);
      \draw (b) edge[arc] (3);
      \draw (c) edge[arc] (3);
      \draw (c) edge[arc] (d);
      \draw (d) edge[arc] (4);
      \draw (d) edge[arc] (5);
      \draw (e) edge[arc] (d);
      \draw (e) edge[arc] (6);
      \draw (f) edge[arc] (7);
      \draw (f) edge[arc] (8);

      \draw[->] (b) edge [compressedSelfLoop] (b);
      \draw (a)  edge[compressed]  (c);
      \draw (d)  edge[compressed, bend left=10]  (6);
      \draw (6)  edge[compressed, bend left=10]  (d);
      \draw (e)  edge[compressed]  (f);
      \node[anchor=mid] at (3.4, -1) { \abbrv{DAG} compression
        $D = (V,A,E)$ of~$\ori{G}$ };

    \end{scope}
    \begin{scope}[yshift=2cm, xshift=-0.5cm]
      \node[sinkVertex] (1) {$\ori{1}$};
      \node[sinkVertex, below of= 1] (2) {$\ori{2}$};
      \node[sinkVertex, above right = 0cm and 1cm of 1] (4) {$\ori{4}$};
      \node[sinkVertex, below of = 4] (6) {$\ori{6}$};
      \node[sinkVertex, below of= 6] (5) {$\ori{5}$};
      \node[sinkVertex, right = 2.5cm of 1] (7) {$\ori{7}$};
      \node[sinkVertex, right = 2.5cm of 2] (8) {$\ori{8}$};
      \node[sinkVertex, below left = 0.1cm and 0.4cm of  1] (3) {$\ori{3}$};

      \draw(1) edge[dedge] (4);
      \draw(1) edge[dedge] (5);
      \draw(2) edge[dedge] (4);
      \draw(2) edge[dedge] (5);
      
      \draw(2) edge[dedgeleft] (3);
      \draw(3) edge[dedgeleft] (2);
      \draw(2) edge[dedgeleft] (1);
      \draw(1) edge[dedgeleft] (2);
      \draw(1) edge[dedgeleft] (3);
      \draw(3) edge[dedgeleft] (1);
      \draw[->] (1) edge [thick, loop above ] (1);
      \draw[->] (2) edge [thick, loop below ] (2);
      \draw[->] (3) edge [thick, loop left ] (3);

      \draw(4) edge[dedgeleft] (6);
      \draw(6) edge[dedgeleft] (4);
      \draw(5) edge[dedgeleft] (6);
      \draw(6) edge[dedgeleft] (5);

      \draw(4) edge[dedge] (7);
      \draw(5) edge[dedge] (7);
      \draw(6) edge[dedge] (7);
      \draw(4) edge[dedge] (8);
      \draw(5) edge[dedge] (8);
      \draw(6) edge[dedge] (8);
      
      \node[anchor=mid] at (1.2, -3) {To-be-compressed $\ori{G} =
        (\ori V, \ori E)$}; 
      \draw[sepline] (4,-3.2) -- (4,1.5); %
    \end{scope}

  \end{tikzpicture}

  \caption{
    A to-be-compressed graph $\ori G$ with vertex set $\ori V = \{\ori
    1,\dots,\ori 8\}$ and a \abbrv{DAG} compression of $\ori G$. The
    compression adds new cluster vertices to $\ori V$, resulting in $V
    = \ori V \cup \{a,\dots, f\}$. The cluster \abbrv{DAG} has the arc
    set~$A$ shown in gray: For instance, the arcs $(a, \ori 
    1)$ and $(b,  a)$ lie in~$A$. For each $v \in V$, the set $\ori C(v)$
    is the set of sinks reachable in $C$ from~$v$, so $\ori C(a) = \{\ori
    1,\ori 2\}$ and $\ori C(b) = \{\ori 1, \ori 2, \ori 3\}$ and $\ori C(\ori 3)
    = \{\ori 3\}$. The compression edges in~$E$ are depicted using
    double lines. A compression edge like $(a,c) \in E$ represents the
    fact that in $\ori G$ there are edges from each vertex in $\ori C(a) =
    \{\ori 1,\ori 2\}$ to each vertex in $\ori C(c) = \{\ori 3, \ori4,\ori
    5\}$; and the compression loop $(b,b) \in E$ implies that $\ori C(b) =
    \{\ori 1, \ori 2, \ori 3\}$ is a clique in~$\ori G$. In total,
    $\ori E = \bigcup_{(u,v) \in E} \ori C(u) \times \ori C(v)$.
  }
  \label{example:DAGCompression}
\end{figure}

It is not immediately obvious that \emph{\abbrv{DAG}s} are the right
data structure to describe the clusters used in graph compressions. Why
not use trees? They are arguably more natural and, more importantly,
can be encoded efficiently -- a desirable property in the context of
data compression. However, \abbrv{DAG} compressions turn out to be the more
flexible data structure as, for instance, given disjoint \abbrv{DAG}
compressions $D_1$ and~$D_2$ of two possibly overlapping graphs $\ori
G_1$ and~$\ori G_2$, the union of $D_1$ and~$D_2$ is a \abbrv{DAG}
compression of the union of $\ori G_1$ and $\ori G_2$. In contrast,
tree compressions do not seem to enjoy such easy composition
properties, although a formal argument showing that \abbrv{DAG}s
allow better graph compressions than trees was missing.

The second reason we investigate \abbrv{DAG} compressions rather than
tree compressions is that ``we get them for free algorithmically,''
that is, the known algorithms that run directly on the compression
work not only for tree-like clusterings, but on arbitrary \abbrv{DAG}
compressions. Concretely, Bannach et al.~\cite{BannachMT24} showed
that the standard graph algorithms \emph{depth-first search,} \emph{breadth-first
search,} \emph{strongly connected components,} \emph{topological
sorting,} and \emph{Dijkstra's algorithm} can be made to work directly on 
\abbrv{DAG}-compressed graphs. Crucially, the implementation has a 
runtime that is $O(s)$ (or $O(s\log s)$ for Dijkstra's algorithm)
where $s$ is the (small) size of the compression, as opposed to the
(large) size of the original graph. Bannach et al.\ also show that
many dense graphs, such as graphs of bounded twin-width, admit
\abbrv{DAG} compressions of a size $s$ that is linear in the number of
vertices of the to-be-compressed graph, rather than in the typically
quadratic number of edges.   

The central question left open in previous work was
how hard it is to find a (size-)optimal compression of a given graph,
that is, the complexity of the following problem:

\begin{problem}[$\Lang{min-dag-compression}$]\label{problem:dag-compression}\hfil
  \begin{description}
    \item[Input:] A directed graph $\ori{G} = (\ori{V}, \ori{E})$, $k\in \mathbb{N}$.
    \item[Question:] Does $\ori G$ have a \abbrv{DAG} compression $D = (V, A, E)$ with
    $|E| + |A| \leq k$? 
  \end{description}
\end{problem}

\subparagraph*{Our Contributions.}

In our paper, we, firstly, enlarge the list of algorithms that
can work directly on \abbrv{DAG}-compressed graphs by another
fundamental entry: We show that Kruskal's algorithm for computing
minimum spanning trees (\abbrv{MST}s) can be implemented so that it works on
\abbrv{DAG} compressions. Formally, we have ($\alpha$ is the inverted Ackermann
function):

\begin{theorem}\label{theorem:kruskal}
  \gdef\textoftheoremkruskal{%
    On input of a weighted \abbrv{DAG} compression $D = (V,A,E,w)$ of a
    weighted connected undirected graph $\ori G = (\ori V, \ori E,\ori
    w)$, we can compute an \abbrv{MST} of $\ori G$ in time $O\bigl(
    (|A|+|E|) \cdot \alpha(|\ori{V}|) + t_{\mathrm{sort}}(|E|)\bigr)$,
    where $t_{\mathrm{sort}}$ is the time needed to sort a weight array.
  }%
  \textoftheoremkruskal
\end{theorem}

If we know nothing about the weights, $t_{\mathrm{sort}}$
is $\Theta(|E| \log |E|)$, but if it is known that the weights
can be, for instance, bucket-sorted, then $t_{\mathrm{sort}}$ may be
linear and can be ignored.


Unlike the algorithms of Bannach et al.~\cite{BannachMT24}, which are ``just''
the standard algorithms but applied to a special ``switching graph'' that
results from a \abbrv{DAG} compression and is ``distance-preserving'' relative
to the uncompressed graph, our version of Kruskal's algorithm needs some
dedicated algorithmic adaptions. In particular, it is the first algorithm
which works on \emph{weighted} DAG compressions. Here, an edge of the input
graph can be represented by multiple compression edges, and the weight of an
edge $\ori e$ in the original graph is defined to be the minimum weight of any
compression edge $e$ representing $\ori e$ in the compression.

A second question we address is whether \abbrv{DAG} compressions
really offer an improvement over the conceptually simpler tree
compressions, which have the extra benefit of allowing us to easily store the
whole tree structure using only $O(|\ori V|)$ bits. We show that the
answer is positive: Consider the \emph{rook graph} $\ori R_{g\times g} :=
\bigl(\{1,\dots,g\} \times \{1,\dots,g\},
\bigl\{\bigl((r_1,c_1),(r_2,c_2)\bigr) \mid r_1 =r_2 \lor c_1 = c_2
\bigr\}\bigr)$, so-called as it corresponds to the possible movements
of a rook in chess. We can trivially \abbrv{DAG}-compress this graph
such that $|A| = 2g^2$ and $|E| = g^2$. However, a non-trivial argument will
show the following:

\begin{theorem}\label{theorem:rook}
  \gdef\textoftheoremrook{%
  Let $T = (V,A,E)$ be a tree compression of the rook graph $\ori
  R_{g\times g}$. Then the number $|E|$ of compression edges is at
  least $g^3/32 - g^2$.}%
  \textoftheoremrook
\end{theorem}
Thus, tree compressions of the $n$-vertex rook graph $\ori
R_{g\times g}$ with $n = g^2$ need $\Omega(n^{3/2})$
compression edges. Since the number of original edges is $|\ori E| =
2 g^3 = 2 n^{3/2}$, this means that no significant compression is
possible using tree-like clusterings. In
contrast, a \abbrv{DAG} compression only needs $O(s)$ edges for $s =
|A| + |E|$ together. All told,
the rook graph can be  compressed better by a factor of $\sqrt n$ using
\abbrv{DAG}s than using trees.

The third question we address and answer is the central problem of how
difficult it is to compute an optimal \abbrv{DAG} compression. We
prove: 

\begin{theorem}\label{theorem:main}
  \gdef\textoftheoremmain{%
    $\Lang{min-dag-compression}$ is $\Class{NP}$-complete.%
  }%
  \textoftheoremmain
\end{theorem}

The proof is via a non-trivial reduction from the
$\Lang{set-cover}$ problem. Interestingly, for \abbrv{DAG}
compressions, ``search does not reduce to decision'', meaning that
even if we get $\Lang{min-dag-compression}$ as an oracle, it is not
clear how this helps us to compute optimal \abbrv{DAG} compressions:
The usual strategy of successively adding edges to a growing partial
solution and querying the oracle as to whether an optimal solution is
maintained may not work as later graphs can have much smaller
compressions than intermediate graphs. This leads us to the
non-trivial questions of how difficult it is to \emph{update a given
optimal \abbrv{DAG} compression when a single edge is added or deleted in the
graph}.\footnote{Of course, if we can decide
$\Lang{min-dag-compression}$ efficiently, then $\mathrm P = \mathrm{NP}$
and constructing optimal \abbrv{DAG} compressions becomes easy as the
problem ``Can some partial \abbrv{DAG} compression be extended to a
size-$k$ \abbrv{DAG} compression of $\ori G$?'' is then
also in $\mathrm{P}$. However, this does not answer the question of
the difficulty of a single update.} We show that the following two problems are both $\Class{NP}$-hard:


\begin{problem}[$\Lang{min-dag-compression-add}$]\label{problem:update-add}\hfil
  \begin{description}
  \item[Input:] A graph $\ori{G} = (\ori{V}, \ori{E})$, a 
    \abbrv{DAG} compression $D=(V,A,E)$ of $\ori G$, promised to have
    minimal size $|A| + |E|$, a new edge $\ori e \in \ori V \times
    \ori V$, and an integer $k\in \mathbb{N}$. 
    \item[Question:] Is there a \abbrv{DAG} compression $D'=(V,A',E')$
    of $(\ori{V},\ori{E} \cup \{ \ori{e} \})$ with $|A'| + |E'| \leq k$? 
  \end{description}
\end{problem}

\begin{problem}[$\Lang{min-dag-compression-delete}$]\label{problem:update-delete}\hfil
  \begin{description}
  \item[Input:] A graph $\ori{G} = (\ori{V}, \ori{E})$, a 
    \abbrv{DAG} compression $D=(V,A,E)$ of $\ori G$, promised to have
    minimal size $|A| + |E|$, an edge $\ori e \in \ori E$, and an integer $k\in \mathbb{N}$. 
    \item[Question:] Is there a \abbrv{DAG} compression $D'=(V,A',E')$
    of $(\ori{V},\ori{E} \setminus \{ \ori{e} \})$ with $|A'| + |E'| \leq k$? 
  \end{description}
\end{problem}






\subparagraph*{Related Work.}

Research of graph compression techniques has a long and fruitful history,
with continued interest to this day~\cite{BlandfordBK03, BoldiV04,
BouritsasLKB21, ChierichettiKLMPR09, DhulipalaKKOPS16, ShunDB15, Versari21}. The
\abbrv{DAG} compression data structure we study in this paper was recently
introduced in~\cite{BannachMT24}.
There are many proposed graph compression methods, one particularly close to
\abbrv{DAG} compressions was introduced in~\cite{ToivonenZHH11}. They utilize
similar ``supernodes'' and ``superedges'', where a superedge between two
supernodes represents all edges between the vertices within these supernodes.
However, their method involves partitioning the vertices into supernodes, in
contrast to \abbrv{DAG} compressions, in which nested clusterings of vertices are
possible. The approach in~\cite{NavlakhaRS08} resembles that of
\cite{ToivonenZHH11}, but it includes additional edge corrections to restore the
original graph.
Another similar method is a visualization tool for
dense graphs, called a \emph{power graph} \cite{DwyerRMM13}, which uses
modules to display all edges between two modules with a single
edge. Nesting of modules is allowed, but if there is any overlap
between two modules, one must be completely contained within the
other, meaning that this tool corresponds to tree compressions in our
parlance. Dwyer et al.~\cite{DwyerMMNMW14} demonstrated that finding a 
minimal power graph with just a single module is $\Class{NP}$-hard, and
suggested that this hardness likely extends to more general cases as well.

\abbrv{DAG} compressions where originally inspired by graphs with
bounded twin-width, a parameter introduced by Bonnet et
al.~\cite{BonnetKTW22}: Such graphs always allow even a tree
compression of linear size. Indeed, \emph{twin models,}
introduced in~\cite{Bonnet23}, are a special case of tree
compressions. Bannach et al.~\cite{BannachMT24} show that
there are graphs (namely, for instance, the earlier-mentioned rook
graph) that have linear-size \abbrv{DAG} compressions, but do not have
bounded twin-width. This suggests and was conjectured
in~\cite{BannachMT24}, but does not imply, that the rook graph has no
linear-size tree compression. 

Concerning the complexity of finding minimal-size \abbrv{DAG}
compressions, Bannach et al.~\cite{BannachMT24} already showed that it
is $\Class{NP}$-hard to decide on input $\ori G$ and a number of~$k$
whether $\ori G$ admits a \abbrv{DAG} compression $(V,A,E)$ with $|E|
\le k$. In other words, finding a compression with a minimal number of 
compression edges is hard -- but, as already pointed out
in~\cite{BannachMT24}, this is ``not really the question'' since any
encoding of $(V,A,E)$ must also encode $A$, so minimizing $|E|$ at the
expense of $|A|$ is only of theoretical interest. Of course,
intuitively, it seems a bit unlikely (though not impossible) that
minimizing $|E|$ alone is hard while minimizing $|A| + 
|E|$ suddenly turns out to be easy; the results 
of the present paper show that proving this intuition to be correct is
surprisingly complex.

\subparagraph*{Organization of this Paper.}

After introducing the necessary terminology in
Section~\ref{section:background}, we devote one section to each
of the three earlier theorems: In Section~\ref{section:kruskal} we
prove Theorem~\ref{theorem:kruskal} by showing how \abbrv{DAG}
compressions can be used to find minimum spanning trees more efficiently. In
Section~\ref{section:dag-vs-tree}, we show that \abbrv{DAG}
compression can yield strictly smaller compressions than tree
compressions by proving a lower bound on the size of any tree
compression of the rook graph, thereby proving
Theorem~\ref{theorem:rook}. Finally, we address the intractability of
computing and updating optimal \abbrv{DAG} compressions in 
Section~\ref{section:hardness} and prove Theorem~\ref{theorem:main}.
Proofs moved to the appendix got replaced by~``$\blacktriangledown$''
in the main text. 

\tcsautomoveaddto{main}{
  \clearpage
  \appendix
  \section{Technical Proofs}
  In the following, we provide the proofs omitted in the main text. In each
  case, the claim of the theorem or lemma is stated once more for the reader's
  convenience. }

\section{Background}
\label{section:background}


In this paper, a (directed) \emph{graph} is a pair $G = (V,E)$,
consisting of a set $V$ of \emph{vertices} and an \emph{edge relation}
$E \subseteq V \times V$. A \emph{weighted} graph is a triple $G =
(V,E,w)$ with $w \colon E \to \mathbb N$. A \emph{twin} in a graph $G
= (V,E)$ is a pair $t_1$,~$t_2$ of vertices such that their in- and
out-neighborhoods are 
identical, that is, $\{v \mid (v,t_1) \in E\} = \{v \mid (v,t_2) \in
E\}$ and $\{v \mid (t_1,v) \in E\} = \{v \mid (t_2,v) \in
E\}$. A \emph{walk of
length~$k$ from $v_1$ to $v_k$} in a 
graph is a sequence $(v_1,\dots,v_k)$ of vertices 
such that $(v_i,v_{i+1}) \in E$ holds for all $i \in
\{1,\dots,k-1\}$; a \emph{path} is a walk in which all vertices are
distinct; a \emph{cycle} is a walk in which all vertices are distinct
except for the first and last, which must be identical. A graph is
\emph{acyclic} if it has no cycle and, as is standard, we call
acyclic graphs \emph{\abbrv{DAG}s} (\emph{directed acyclic graphs}). A
\emph{tree} is a \abbrv{DAG} in which there is a unique \emph{root} $r
\in V$ such that for all vertices $v \in V$ there is a unique path
from $r$ to~$v$. A
\emph{forest} is a vertex-disjoint union of trees. 
A \emph{sink} in a \abbrv{DAG} is a vertex $v \in V$
of out-degree~$0$, that is, for which there is no vertex $u$ with
$(v,u) \in E$. The sinks of a tree are also called \emph{leaves.}

Recall the definition of \abbrv{DAG} compressions from the
introduction: We start with a ``normal, typically dense'' graph $\ori
G = (\ori V, \ori E)$, which we denote in bold face to better
distinguish it from the vertices and edges used for the
compression. We compress it using a \abbrv{DAG}
compression~\cite{BannachMT24}, which is 
a triple $D = (V, A, E)$ consisting of:
\begin{enumerate}
\item A \emph{cluster \abbrv{DAG}} $(V, A)$ such that the set of
  its sinks is exactly $\ori V$. The  set~$A$ is simply the edge
  relation of the cluster \abbrv{DAG}, but we will call the edges
  in~$A$ \emph{arcs} to better distinguish them verbally from the
  normal edges in~$\ori E$. We  associate a 
  \emph{cluster $\ori C_D(v) \subseteq \ori V$} (or just $\ori C(v)$
  when $D$ is clear from context) with each \emph{cluster vertex} $v
  \in V$ by setting $\ori C(v) := \{   u\in V \mid u$ is reachable from 
  $v\}$. Observe that $\ori C(\ori v) = \{\ori v\}$ holds for 
  all $\ori v \in \ori V \subseteq V$, so each singleton
  set of original vertices is available as a cluster.
\item
  A set $E \subseteq V \times V$  of \emph{compression edges} with
  $\ori E
  = \bigcup_{(u,v) \in E} \ori C(u) \times \ori C(v)$. Each
  compression edge $(u,v) \in E$ between two cluster 
  vertices encodes the fact that in $\ori G$ there are edges from all
  vertices in the cluster $\ori C(u)$ to all vertices in the cluster
  $\ori C(v)$. If the two clusters are
  disjoint, then $\ori C(u) \times \ori C(v)$ is a \emph{biclique} (a
  complete bipartite graph). If the clusters overlap, we call
  $\ori C(u) \times \ori C(v)$ a \emph{product (of $u$'s and $v$'s
  clusters)}.
\end{enumerate}
We define the \emph{size} of a \abbrv{DAG} compression $D = (V, A, E)$
as $|A| + |E|$. (The reason we do not include $|V|$ in the size is
that $V \setminus \ori V$ cannot contain sinks and, thus,
all ``costs of encoding the cluster \abbrv{DAG} are fairly paid for 
by including~$|A|$''.)

The above definitions can easily be adapted to \emph{undirected
graphs,} which are pairs $\ori G = (\ori V, \ori E)$, where $\ori V$ is a set
of vertices and $\ori E \subseteq \{\{\ori u,\ori v\} \mid \ori u,\ori
v \in \ori V\}$ is a set of undirected edges. An 
\emph{undirected tree} is an undirected graph in which there is a unique path
from a root vertex to every other vertex. For two sets $A$ and $B$, let $A
\otimes B$ denote the set $\{\{a, b\} \mid {a\in A}, {b\in B}\}$. An
\emph{undirected \abbrv{DAG} compression} is then a triple $D = (V, A, E)$
consisting of a (still directed) \emph{cluster \abbrv{DAG}} $(V, A)$, such that the set
of its sinks is exactly~$\ori V$, and a set $E
\subseteq V \otimes V$  of \emph{undirected compression edges} with $\ori E =
\bigcup_{\{u,v\} \in E} \ori C(u) \otimes \ori C(v)$, where clusters and cluster
vertices are defined as in the directed case.

In the context of minimum spanning trees, the inputs are
\emph{weighted undirected graphs,} which are triples $\ori G = (\ori
V,\ori E,\ori w)$ with $\ori w \colon \ori E \to \mathbb N$. An
\emph{(undirected) weighted \abbrv{DAG} compression of~$\ori{G}$} is 
a quadruple $D=(V,A,E,w)$, such that $(V,A,E)$ is a \abbrv{DAG} compression
of~$(\ori{V},\ori{E})$ and such that for every $\{\ori x, \ori y\} \in
\ori{E}$ we have that $\ori{w}(\{\ori x, \ori y\}) = \min_{\{u,v\} \in
  E, \ori x \in C(u), \ori y \in C(v)} w(\{u,v\})$.

\subparagraph*{Finding Sinks in Constant Time.}
For some of our algorithms it will be useful to quickly obtain on input
of a cluster vertex $v \in V$ ``some arbitrary element of $\ori
C(v)$,'' that is, some sink that is reachable from $v$ in
$(V,A)$. While this is easy enough to achieve in principle, in
order to perform this operation in time~$O(1)$, some initial
preprocessing is needed: 
\begin{lemma~}\label{lemma-representation}
  On input of a \abbrv{DAG} compression $D = (V,A,E)$ of $\ori G =
  (\ori V, \ori E)$, we can compute in time
  $O(|A| + |V|)$ a function $\ori c \colon V \to
  \ori V$ with $\ori c(v) \in \ori C(v)$ for all $v\in V$.
\end{lemma~}
\begin{proof~}
  Compute a topological sorting $<$ of $(V,A)$ such that $(u,v) \in A$
  implies $u < v$; so the smallest vertex with respect to $u$ is a
  source and the largest is a sink. Iterate over $V$ in descending
  order with respect to $<$ and set $\ori c(v) := v$, if $v$ is a
  sink, and otherwise set $\ori c(v) := \ori c(u)$, where $u$ is any
  vertex with $(v,u) \in A$. Clearly, we always have $\ori c(v) \in
  \ori C(v)$ since $\ori c(v)$ is always set to a vertex in $\ori V$
  that is reachable from $v$ in $(V,A)$. The runtime follows from the
  standard upper bound on the time needed for topological sorting.
\end{proof~}

\section{Computing Minimum Spanning Trees on DAG Compressions}
\label{section:kruskal}
\tcsautomoveaddto{main}{\subsection{Proofs for Section~\ref{section:kruskal}}}

Suppose we are given a \abbrv{DAG} compression $D = (V,A,E)$ of a graph $\ori G
= (\ori V, \ori E)$ and wish to solve a standard problem like, say, computing
the strongly connected components of~$\ori G$. We could, of course, run a
depth-first search (\abbrv{DFS}) on~$\ori G$ by first uncompressing~$D$, but
this defeats the purpose of compressing graphs in the first place.
Instead, we need to implement \abbrv{DFS} in such a way that it works
\emph{directly on $D$ without ever computing $\ori G$} -- and, preferably, the
runtime should be linear in the size of $D$, rather than in the number of edges
of~$\ori G$. As shown by Bannach et~al.~\cite{BannachMT24}, such
implementations are possible several algorithms, including~\abbrv{DFS}.

In the following, we add another standard algorithm to this list, namely a
version of Kruskal's algorithm for computing minimum spanning trees that works
directly on a weighted undirected \abbrv{DAG} compression
$D=(V,A,E,w)$ of a weighted connected undirected graph $\ori G = (\ori
V,\ori E,\ori w)$. Recall that
a \emph{spanning tree $\ori T = ( \ori V,  \ori E')$} of~$\ori G$ is an
undirected tree with $\ori E'\subseteq  \ori E$. The \emph{weight of~$\ori T$}
is $\sum_{\ori e \in \ori E'} \ori w(\ori e)$. A \emph{minimum
spanning tree} (\abbrv{MST}) is a spanning tree of minimal
weight. Clearly, only connected graphs can have spanning trees. For
unconnected graphs, a \emph{(minimum) spanning forest of $\ori G$} is
an undirected forest in which the trees are minimum spanning
trees of the connected components of~$\ori G$. Kruskal's
algorithm~\cite{Kruskal56} (for uncompressed graphs) is given in
Algorithm~\ref{algo-kruskal}; see
Section~\ref{section:kruskal-appendix} for more background on
union-find data structures.
The key invariant that is upheld during the main loop is: 
\begin{lstlisting}[
    style=pseudocode,
    backgroundcolor=,
    float=ht,
    label={algo-kruskal},
    backgroundcolor=,
    caption={{
        Kruskal's algorithm for computing a minimum spanning tree
        $(\ori V,\ori M)$ of a
        connected undirected graph $(\ori V, \ori E)$, where the individual 
        edges are $(\ori u_1,\ori v_1)$ to $(\ori u_m,\ori v_m)$. The global variable $\ori
        M$ stores the currently computed part of the solution, while
        the global $P$ is the union-find data structure currently holding a
        partition of $\ori V$ so that the connected components of $(\ori
        V,\ori M)$ are exactly the sets in~$P$.
    }}]
algorithm $\Algo{kruskal}(\ori V, \ori E = \{\{\ori u_1,\ori v_1\},\dots, \{\ori u_m,\ori v_m\}\}, \ori w)$
   sort $\ori E$ according to weight so that f${}$or $i<j$ we have $\ori w(\{\ori u_i,\ori v_i\}) \le \ori w(\{\ori u_j,\ori v_j\})$
   $\ori M \gets \emptyset$ 
   $P \gets \Algo{initialize-union-find}(\ori V)$
   foreach $i \in \{1,\dots,m\}$ do
      call $\Algo{add-edge}(\ori u_i,\ori v_i)$
   return $\ori M$

algorithm $\Algo{add-edge}(\ori u,\ori v)$
   if $\Algo{find}(\ori u) \neq \Algo{find}(\ori v)$ then 
      $\ori M \gets \ori M \cup \{\{\ori u,\ori v\}\}$
      $P.\Algo{unite}(\ori u,\ori v)$
\end{lstlisting}
\begin{invariant}\label{inv-kruskal}
  $(\ori V,\ori M)$ is a minimum spanning forest of $(\ori
  V,\{\ori e_1,\dots, \ori e_i\})$.
\end{invariant}
\subparagraph*{Kruskal's Algorithm on \abbrv{DAG}-Compressed Graphs.} The idea behind our version of Kruskal's algorithm for
\abbrv{DAG}-compressed graphs is easy enough: Just iterate over the
compression edges instead of the normal edges and ``somehow'' handle
the edges represented by a compression edge efficiently, see
Algorithm~\ref{algo-kruskal-dag}. 

\begin{lstlisting}[
    style=pseudocode,
    float=ht,
    label={algo-kruskal-dag},
    backgroundcolor=,
    caption={{
        Adaption of Kruskal's algorithm to \abbrv{DAG}-compressed
        graphs. Note that $P$ is a partition of~$\ori V$ (not of~$V$).
    }}]
algorithm $\Algo{kruskal-dag-compressed}(V,A,E = \{\{u_1,v_1\},\dots, \{u_m,v_m\}\},w)$
   sort $E$ according to weight so that f${}$or $i<j$ we have $w(\{u_i,v_i\}) \le  w(\{u_j, v_j\})$
   $\ori M \gets \emptyset$ 
   $P \gets \Algo{initialize-union-find}(\ori V)$ // $\ori V$ are the sinks of $(V,A)$
   foreach $i \in \{1,\dots,m\}$ do
      call $\Algo{add-edges-represented-by}(u_i,v_i)$ // implemented in Algorithm $\ref{algo-edges}$ below
   return $\ori M$
\end{lstlisting}

The analogue to Invariant~\ref{inv-kruskal}, for which we now need to
show that it holds during the main loop of
Algorithm~\ref{algo-kruskal-dag}, would now be \emph{$\ori 
M$ is a  minimum spanning forest of $\bigl(\ori V,(\ori C(u_1) \otimes \ori
C(v_1)) \cup \cdots \cup (\ori C(u_i) \otimes \ori C(v_i)) \bigr)$.}
However, we will show that the following slightly different invariant
holds, which will give us a bit more flexibility in our proofs:
\begin{invariant}\label{inv-kruskal-dag-compressed}
  $(\ori V,\ori M)$ is a minimum spanning forest of $\bigl(\ori
  V, \ori E')$, where $\ori E'$ is some set with $\ori E \supseteq \ori E' \supseteq (\ori C(u_1) \otimes
  \ori C(v_1)) \cup \cdots \cup (\ori C(u_i) \otimes \ori C(v_i))$. 
\end{invariant}
In other words, our invariant just states that we always store a 
minimum spanning forest of an edge set that encompasses \emph{at least} all
uncompressed edges processed during the first $i$ iterations and \emph{at most}
all uncompressed edges. Thus, at the end, we have a minimum spanning
forest of the whole graph $(\ori V, \ori E)$ and hence a minimum spanning tree.

The obvious problem with implementing $\Algo{add-edges-represented-by}$ is
that each compression edge $(u,v)$ corresponds to a whole set
$\ori C(u) \otimes \ori C(v)$ of original edges (up to $\binom{n}{2}+n$ many) and
we may not be free to choose which should be added 
to the \abbrv{MST} as some parts of $\ori C(u)$ and of $\ori C(v)$ may
already be part of larger sets $\ori U \in P$ -- or not. To
complicate things further, $\ori C(u)$ and $\ori C(v)$ may intersect
and may only partly intersect some of the $\ori U \in P$.
Indeed, even just computing $\ori C(u)$ and $\ori C(v)$ for each edge
$\{u,v\} \in E$ is too time consuming. To address these problems, we use
a simple definition:
\begin{definition}
  A vertex $v \in V$ is \emph{clean} if $\ori C(v)
  \subseteq \ori U$ for some $\ori U\in P$.
\end{definition}
In other words, the cluster of a clean vertex must be completely
contained in one of the sets of the partition~$P$. Suppose we had a
way of easily ensuring that a vertex becomes clean. Then implementing
$\Algo{add-edges-represented-by}$ is easy, see Algorithm~\ref{algo-edges}.
\begin{lstlisting}[
    style=pseudocode,
    float=htpb,
    label={algo-edges},
    backgroundcolor=,
    caption={{
        Handling all edges in $\ori C(u) \otimes \ori C(v)$. Recall
        that $\ori c(v)$ is an arbitrary vertex in $\ori C(v)$ that we
        can compute in time~$O(1)$ by
        Lemma~\ref{lemma-representation}.
    }}]
algorithm $\Algo{add-edges-represented-by}(u,v)$
   call $\Algo{make-clean}(u,\ori c(v))$
   call $\Algo{make-clean}(v,\ori c(u))$
   call $\Algo{add-edge}(\ori c(u),\ori c(v))$ 
\end{lstlisting}

Assume for the moment that the two $\Algo{make-clean}$ calls
ensure that both $u$ and $v$ are clean when $\Algo{add-edge}$ is
called. Then the following lemma shows that the call is correct:
\begin{lemma~}\label{lemma-add-edge}
  Suppose Invariant~\ref{inv-kruskal-dag-compressed} holds for
  $i-1$ and we execute
  $\Algo{add-edge}\penalty0(\ori c(u_i),\ori c(v_i))$ for clean vertices $u_i$
  and~$v_i$. Then the invariant will still hold for~$i$ afterwards.
\end{lemma~}
\begin{proof~}
  The call checks whether $P.\Algo{find}(\ori c(u_i)) =
  P.\Algo{find}(\ori c(v_i))$ holds. First suppose this is the
  case. Then $\ori c(u_i)$ and $\ori c(v_i)$ are in the 
  same set $\ori U$ of the partition; and, because of the cleanliness
  of $u_i$ and~$v_i$, so are all other vertices in $\ori C(u_i)$ and
  in $\ori C(v_i)$. In particular, no edge in $\ori C(u_i) \otimes \ori
  C(v_i)$ is between vertices that are not yet in 
  the same set~$\ori U$ of the partition and, hence, they can all be
  skipped. Second, suppose this is not the case. Then we will safely
  add $\{\ori c(u_i), \ori c(v_i)\} \in \ori C(u_i) \otimes \ori C(v_i)$
  to~$\ori M$ and call $\Algo{unite}$, which will unite $P(\ori
  c(u_i))$ and $P(\ori c(v_i))$. Again, because of the cleanliness, in
  the new partition, all of $\ori C(u_i)$ and of $\ori C(v_i)$ will
  end up in the same set as $\ori c(u_i)$ and~$\ori c(v_i)$.
\end{proof~}
Note that calling $\Algo{unite}$ will not cause any clean vertices to
loose that status. 

Of course, not all vertices are clean at the beginning: Indeed,  at
the beginning of $\Algo{kruskal-dag-compressed}$, let us
initialize $v.\mathit{clean}$ to $\mathit{true}$ only for $v \in \ori
V$ and to $\mathit{unknown}$ for all $v \in V \setminus \ori 
V$. Fortunately, there is an easy recursive way of making a vertex
clean when we are in the process of processing a compression edge
$\{u,v\} \in E$, see Algorithm~\ref{algo-make-clean} for the implementation and
Figure~\ref{example:mstComp} on page~\pageref{example:mstComp} for an example.

\begin{lstlisting}[
    style=pseudocode,
    float=htpb,
    label={algo-make-clean},
    backgroundcolor=,
    caption={{
        Recursion for ensuring that vertices are clean while
        adding a compression edge. 
    }}]
algorithm $\Algo{make-clean}(v,\ori r)$
   // Preconditions: $\ori r \in \ori V$ and $\ori C(v) \otimes \{\ori r\} \subseteq \ori E$
   if $v.\mathit{clean} \neq \mathit{true}$ then
      foreach $w\in V$ with $(v,w) \in A$ do
         call $\Algo{make-clean}(w,\ori r)$ $\label{line-mc}$
         call $\Algo{add-edge}(\ori c(w),\ori r)$
      $v.\mathit{clean} \gets \mathit{true}$
\end{lstlisting}

\begin{lemma~}\label{lemma-algo-2}
  Suppose Invariant~\ref{inv-kruskal-dag-compressed} holds for 
  $i-1$ and we execute $\Algo{make-clean}\penalty0(v_i,\ori c(u_i))$. Then 
  $v_i$ will be clean afterwards and the invariant will still hold for
  $i-1$. 
\end{lemma~}
\begin{proof~}
  Let $\ori E'$ be the set of already-spanned edges from the invariant.
  The proof is by structural induction.  We only need to show something
  when $v$ is not yet marked as clean. Consider each child $w \in V$
  of~$v$, meaning $(v,w) \in A$. If $w$ is not yet clean, we call
  $\Algo{make-clean}(w,\ori r)$ in line~\ref{line-mc} and, by the
  induction hypothesis, this will ensure that $w$ is clean (and note
  that the precondition is still satisfied in the recursive call). Consider
  the call $\Algo{add-edge}(\ori c(w),\ori r)$: Since both $w$ and
  $\ori r$ are now clean ($\ori r$ is automatically clean as a sink), we already argued in
  Lemma~\ref{lemma-add-edge} that we will correctly add an edge (if
  necessary) to~$\ori M$ so that $\ori M$ is a  spanning
  forest of $\ori E' \cup (\ori C(w) \otimes \{\ori r\})$.

  The crucial observation is that at the end of the loop, all children
  $w \in V$ of $v$ are clean \emph{and} $\ori M$ will be a 
  spanning forest of $\ori E'' := \ori E' \cup \bigcup_{(v,w) \in A}
  (\ori C(w) \otimes \{\ori r\})$. However, since $\ori r$ is connected to all
  vertices in all~$\ori C(w)$ in $\ori E''$, all $\ori C(w)$ lie in
  the same connected component of $\ori E''$. By definition, in the
   spanning forest $\ori M$, they must also lie in the same
  connected component. Since the connected components of $\ori M$ are
  exactly the sets in~$P$, we conclude that there is a single $\ori U \in
  P$ with $\bigcup_{(v,w) \in A} \ori C(w) \subseteq \ori U$. Since $\ori
  C(v) = \bigcup_{(v,w) \in A} \ori C(w)$, we conclude that $\ori
  C(v)$ lies completely in some $\ori U \in P$ and it is, thus, correct to
  declare $v$ as clean in the last line.

  Furthermore, in each iteration, the spanning forest is a minimum spanning
  forest, since the edges are sorted ascending according to their weight.
\end{proof~}

All told, we get Theorem~\ref{theorem:kruskal} from the introduction:

\begin{claim*}[of Theorem~\ref{theorem:kruskal}]
  \textoftheoremkruskal
\end{claim*}

\begin{proof}
  We run Algorithm~\ref{algo-kruskal-dag}, whose correctness follows
  from Lemmas \ref{lemma-add-edge} and~\ref{lemma-algo-2}. It remains
  to argue that the runtime is correct: First, by
  Lemma~\ref{lemma-representation}, the precomputation of $\ori c$
  takes time $O(|A| + |E| + |\ori V|)$, which is $O(|A| + |E|)$ as the
  graph is connected. Second, observe that during the whole run of the
  algorithm, in any call of Algorithm~\ref{algo-make-clean}, whether
  directly or through the recursion, no edge in $A$
  is processed more than once, causing one call of 
  $\Algo{add-edge}$. Similarly, each edge $\{u_i,v_i\} \in E$ is also
  processed only once and causes one call of $\Algo{add-edge}$. Since
  $\Algo{add-edge}$ takes amortized time $\alpha(|\ori V|)$, we get
  the claimed runtime.
\end{proof}

\section{Lower Bounds on the Size of Tree Compressions}
\label{section:dag-vs-tree}
\tcsautomoveaddto{main}{\subsection{Proofs for Section~\ref{section:dag-vs-tree}}}

Tree compressions are \abbrv{DAG} compressions $(V,A,E)$ where $(V,A)$
is a tree. Besides being conceptually simpler, they also allow more
efficient encodings ($n$-vertex trees can easily be encoded using
$O(n)$ bits), leading to the question of whether \abbrv{DAG}
compressions offer any advantage over tree compressions. Earlier
work~\cite{BannachMT24} conjectured that this might be the case (and
the intuition strongly suggests it), but no proof was found. We fix
this now by proving Theorem~\ref{theorem:rook} in the following, which
states: \emph{\textoftheoremrook}
Recall that the $n$-vertex rook graph $\ori R_{g \times g}$ has a grid
of size $g \times g$ with $g = \sqrt{n}$ as its vertices and two
vertices are connected iff they are in the same row or the same
column. It is easy to construct a \abbrv{DAG} compression $(V,A,E)$ of
$\ori R_{g \times g}$ of size $|A| + |E| = 2g^2 + 2g$, namely by
introducing a row cluster vertex $r_i$ for $i\in\{1,\dots,g\}$ and a
column cluster vertex $c_i$ and to connect in $A$ each $r_i$ to all
vertices of the grid in row~$i$, to connect each $c_i$ to all vertices
in column~$c_i$, and to put compression loops $(r_i,r_i)$ and
$(c_i,c_i)$ into~$E$ for all $i \in \{1,\dots,g\}$ to represent the
row and column cliques. All told, the rook graph
admits a \abbrv{DAG} compression of size~$O(n)$, while
Theorem~\ref{theorem:rook} states that any tree compression has size $|A| + |E| \in
\Omega(n^{3/2})$.

For the proof, we need a lemma:
\begin{lemma~}\label{lemma-threes}
  For $\ori R_{g\times g} = (\ori V, \ori E)$ let $\ori X, \ori X'
  \subseteq \ori V$ be sets with $|\ori X| \ge 3$ and $|\ori X'| \ge 3$ and
  $\ori X \times \ori X' \subseteq \ori E$. Then all vertices in $\ori
  X \cup \ori X'$ lie on the same row or the same column. 
\end{lemma~}
\begin{proof~}
  Assume the conclusion does not hold. Now assume that, at
  least, the vertices in~$\ori X$ lie on the same 
  row $\ori r$. Then there must be an $(\ori r',\ori c') \in \ori
  X'$ with $\ori r \neq \ori r'$. But, then, as all vertices in
  $\ori X$ must lie in different columns (they lie in the same row),
  $(\ori r',\ori c')$ would be connected to at least two vertices
  that lie both in a different row (namely $\ori r$) and different
  columns. By a similar argument, the vertices in $\ori X$ also cannot
  all lie on the same column.

  When three vertices neither all lie on a row nor 
  all on a column, two of them must lie both on different rows and
  different columns, that is, there must be $(\ori r_1,\ori c_1), (\ori
  r_2,\ori c_2 )\in \ori X$ with $\ori r_1 \neq \ori r_2$ and $\ori
  c_1 \neq \ori c_2$.
  Now consider any vertex $(\ori r',\ori c') \in\ori X'$. As
  it is connected to  all vertices in~$\ori X$, it must share a row or a
  column with $(\ori r_1,\ori c_1)$ and also with $(\ori r_2,\ori
  c_2)$. This is only possible either for $\ori r'=\ori r_1$ and $\ori
  c'=\ori c_2$ or for $\ori r'=\ori r_2$ and $\ori c'=\ori c_1$ (other
  cases are ruled out by $\ori r_1 \neq \ori r_2$ and $\ori
  c_1 \neq \ori c_2$). But this means that there are only \emph{two}
  possibilities for $(\ori r',\ori c')$, contradicting $|\ori X'| > 2$. 
\end{proof~}

\begin{proof}[Proof of Theorem~\ref{theorem:rook}]
  Let $T = (V, A, E)$ be a tree compression of the rook graph $\ori
  R_{g\times g}$ such that $|E|$ is minimal. Let $\ori V$ denote the
  vertex set of $\ori R_{g  \times g}$, that is, the $g \times g$
  grid. We may assume that the 
  cluster tree $T = (V, A)$ is a binary tree, meaning that each
  cluster vertex $v \in V$ either has exactly two children or is a
  leaf (and then an element of $\ori V$), since we can easily adapt
  $A$ to satisfy this condition without changing~$E$.

  Let us call a cluster vertex $v \in V$ \emph{big} if $|\ori C(v)|
  \ge 3$. A big vertex is \emph{horizontal} if all vertices in $\ori
  C(v)$ lie in the same row, and \emph{vertical} if they all lie in
  the same column. A big vertex that is neither horizontal nor
  vertical is a \emph{cross} if there is pair $(\ori r_0,\ori c_0)$,
  called the \emph{crosshair of~$v$}, such that all $(\ori r,\ori c)
  \in \ori C(v)$ have $\ori r=\ori r_0$ or $\ori c=\ori c_0$. Note
  that when $v$ is a parent of $u$ in the tree and both $v$ and $u$
  are crosses, then $v$ and $u$ have the same
  crosshair. Figure~\ref{RookGraph} depicts an example.
\begin{figure}[th]
  \centering
  \begin{tikzpicture}
    
    \tikzset{every node/.style={node},
      treecompressionnode/.style={clusterVertex},
      treecompression/.style={
        arc
    }}

    \foreach \x in {0,1,2,3, 4, 5}  {
      \draw (0, \x) -- (5, \x);
      \draw (\x, 0) -- (\x, 5);
    }
    \foreach \x in {0,1,2,3, 4, 5}
    \foreach \y in {0,1,2,3, 4, 5} {
      \node[node,thick] (t\x\y) at (\x,\y) {};
    }


    \node[treecompressionnode] (t1-1) at (-1,1.5) {};
    \node[treecompressionnode] (t1-2) at (-1.5,2.5) {$v_1$};
    \draw[arc, treecompression] (t1-2) -- (t1-1);
    \draw[arc, treecompression] (t1-2) -- (t03);
    \draw[arc, treecompression] (t1-1) -- (t02);
    \draw[arc, treecompression] (t1-1) -- (t01);

    \node[treecompressionnode] (t2-1) at (-2.75, 2.5) {$v_2$};
    \draw[arc, treecompression] (t2-1) -- (t04);
    \draw[arc, treecompression] (t2-1) -- (t00);


    \node[treecompressionnode] (t3-1) at (1.5,5.5) {};
    \node[treecompressionnode] (t3-2) at (2.1,6.3) {$h$};
    \draw[arc, treecompression] (t3-2) -- (t3-1);
    \draw[arc, treecompression] (t3-2) -- (t34);
    \draw[arc, treecompression] (t3-1.-85) -- (t24);
    \draw[arc, treecompression] (t3-1.-95) -- (t14);

    \node[treecompressionnode] (t4-1) at (4.5, 5.5) {$s$};
    \draw[arc, treecompression] (t4-1.-95) -- (t44);
    \draw[arc, treecompression] (t4-1.-85) -- (t54);


    \node[treecompressionnode] (t5-1) at (6, 1.2) {};
    \node[treecompressionnode] (t5-2) at (7, 0.5) {$c$};
    
    \draw[arc, treecompression] (t5-1) -- (t42);
    \draw[arc, treecompression] (t5-1) -- (t51);
    \draw[arc, treecompression] (t5-2) -- (t5-1);
    \draw[arc, treecompression] (t5-2) -- (t50);

    \draw[compressedUn] (t1-2) -- (t2-1);
    \draw[compressedUn] (t3-2) -- (t4-1);

    \draw[compressedUn] (t5-2) to[bend right=50] (t52);

  \end{tikzpicture} 

  \caption{Example for part of a tree compression of the rook graph
    with $g = 6$. For legibility, the row cliques and column cliques
    are only indicated by the edges in~$\ori E$ between direct neighbors in
    the rook graph. The vertices $v_1$ and $v_2$ are vertical, the
    vertex $h$ is horizontal, but $s$ is not horizontal as it is not
    big. The cross vertex $c$ is connected by a 
    compression edge to its crosshair.
  }
  \label{RookGraph}

\end{figure}
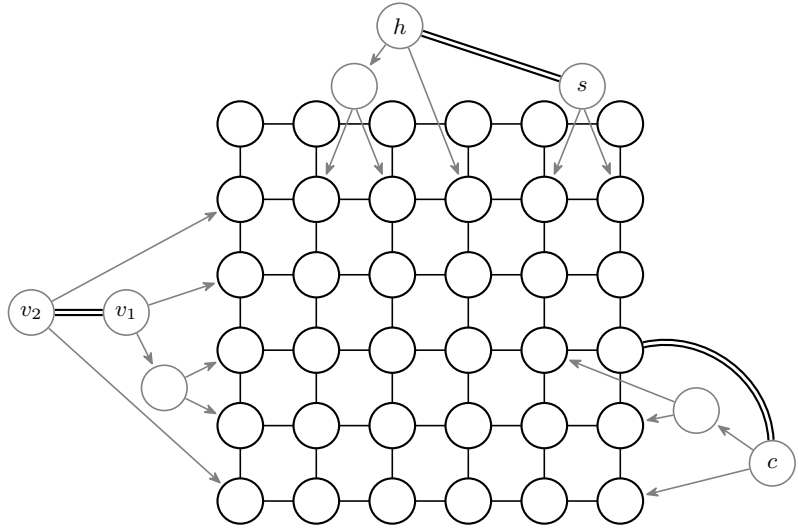

  We make some easy observations:
  \begin{enumerate}
  \item If $(u,v) \in E$ is a
    compression edge connecting two big vertices, then all vertices in
    $\ori C(u) \cup \ori C(v)$ must lie on a single row or in a single
    column (this is a direct consequence of
    Lemma~\ref{lemma-threes}).
  \item If $(u,v) \in E$ connects a
    cross~$u$ and a vertex~$v$, then $v$ must be the crosshair of~$u$  
    (since $u$ is neither horizontal nor vertical, it contains vertices
    from both different rows and different columns and can thus only be 
    connected to the crosshair). Furthermore, we may assume that the
    parent~$p$ of~$u$ is not a cross (since we would get only a
    potentially smaller tree compression by replacing $(u,v)$ by $(p,v)$
    in~$E$; recall that $u$ and~$p$ must have the same crosshair).
  \item
    There can be no $(u,v) \in E$ where $u$ is big, but not even a
    cross. 
  \end{enumerate}
  Let us now count the number of edges in $\ori E$ that can be covered
  by the compression edges in~$E$. We first consider edges $(u,v) \in E$ connecting
  a cross $u$ and its crosshair~$v$: Such a compression edge
  can represent at most $|\ori C(u)|$ many edges in $\ori
  E$. Crucially, as we argued that $u$ will be a maximal cross (and
  there cannot be crosses further up towards the root: once we loose
  the status of being a cross, we cannot regain it), the sets $\ori
  C(u)$ are pairwise disjoint for maximal crosses~$u$. This means that
  the sum of their sizes is at most $|\ori V| = g^2$. All told,
  compression edges $(u,v)$ involving crosses can represent at most
  $g^2$ edges in~$\ori E$.

  Next, for a given row $\ori r \in \{1,\dots, g\}$, consider all
  vertices $(\ori r,\ori c)$ for $\ori c \in \{1,\dots,g\}$ in the
  row and consider the first big vertex $v$ that is an ancestor of
  $(\ori r,\ori c)$. Intuitively, if $v$ is a horizontal vertex, this
  is helpful for compressing the edges in the row, so let us count the
  number of vertices $v$ for which this is \emph{not} the case and
  call this number~$\bar h_{\ori r}$.
  In the same way, define $\bar v_{\ori c}$ as the number of vertices
  in column $\ori c$ whose first big ancestor is \emph{not} a vertical
  vertex. 
  Since the first big ancestor of a vertex $(\ori r,\ori c)$ cannot
  be both horizontal and vertical at the same time, each vertex $\ori
  V$ counts in some $\bar h_{\ori r}$ or $\bar v_{\ori c}$. In
  particular, $\sum_{\ori r \in \{1,\dots, g\}} \bar h_{\ori r} +
  \sum_{\ori c \in \{1,\dots, g\}} \bar v_{\ori c} \ge g^2$. This
  implies that at least one of the sums is $g^2 /2$, say the first
  (the other case is symmetric).

  We now have $\sum_{\ori r \in \{1,\dots, g\}} \bar h_{\ori r} \ge g^2/2$. This
  means that a sum of $g$ numbers, each at most~$g$, is at least $g^2/2$. This
  is only possible when at least half of the numbers are at least $g/2$, so for
  at least $g/2$ rows we have $\bar h_{\ori r} \ge g/2$. Now consider such a
  row: At least $g/2$ vertices $\ori v$ do \emph{not} have a horizontal vertex
  as their first big ancestor. Consider the at least $g^2/4$ edges in $\ori E$
  between them and how they can be represented by compression edges: They
  can\emph{not} be represented by compression edges between two big vertices
  since these would have to be horizontal vertices, but none of the ancestors of
  the~$\ori v$ are horizontal. They also cannot be represented by a compression
  vertex involving a big edge that is not a cross. We are left with compression
  edges between non-big vertices in $V$ and between a cross and its crosshair.

  In total, we get that we need to cover \emph{at least $(g^2/4)\cdot
  (g/2)= g^3/8$ edges in $\ori E$ using compression edges between
  non-big cluster vertices and between crosses and their crosshairs.}
  We argued already the compression edges between crosses and their
  crosshairs can contribute at most $g^2$ edges. Now, since a
  compression edge between two small cluster vertices can contribute
  at most $4$~edges, \emph{there must be at least $(g^3/8)/4 - g^2 =
  g^3/32 - g^2$ compression edges in~$E$.}
\end{proof}

By the above results, a \abbrv{DAG} compression of an $n$-vertex graph
can have size $O(n)$ while the best tree compression has size
$\Theta(n^{3/2})$, meaning the compression is better by a factor
of~$\sqrt n$. Can we do better? Since a graph can have up to $O(n^2)$ 
edges, the theoretical maximum is a factor of~$n$. We believe that we
can get arbitrarily close to that factor by adapting the argument for
the rook graph to rook graphs in higher dimensions: For instance, let
$\ori R_{g\times g \times g}$ be the graph whose vertices form a three
dimensional grid and where there is an edge between $(\ori x,\ori
y,\ori z)$ and $(\ori x',\ori y',\ori z')$ iff $\ori x = \ori x'$ or
$\ori y = \ori y'$ or $\ori z = \ori z'$. Again, it is easy to find a
\abbrv{DAG} compression of size $O(g^3) = O(n)$ of this graph, but we
believe that our argument for $\ori R_{g\times g}$ can be adapted to
show that any tree compression of $\ori R_{g\times g \times g}$ has size $\Omega(g^5) = n^{5/3}$. However,
the argument does not seem to be straightforward, so we formulate a
conjecture: 
\begin{conjecture}
  Any tree compression of the $d$-dimensional rook graph $\ori
  R_{g^d}$ has size $\Omega(n^{(2d-1)/d})$.
\end{conjecture}

\section{Hardness of Computing and Updating DAG Compressions}
\label{section:hardness}
\tcsautomoveaddto{main}{\subsection{Proofs for Section~\ref{section:hardness}}}

Both previous work and the earlier results show that \abbrv{DAG}
compressions offer a way of efficiently running algorithms on large
dense graphs. Naturally, we first need to obtain a \abbrv{DAG}
compression in the first place. For some applications, where graphs
are generated algorithmically, this may be easy to do, but in general
we get a graph $\ori G$ as input and need to compute an as-small-as-possible
\abbrv{DAG} compression $(V,A,E)$ of~$\ori G$. Unfortunately, 
in the following we prove Theorem~\ref{theorem:main}, which states \emph{\textoftheoremmain}

The proof is by a reduction from the $\Lang{set-cover}$ problem. Given
a set $T = \{S_1,\dots,S_n\}$, recall that a \emph{cover} of a set~$V$
is a subset $X \subseteq T$ such that $V \subseteq \bigcup
X$.
\begin{problem}[$\Lang{set-cover}$]\label{problem:set-cover}\hfil
  \begin{description}
    \item[Input:] A universe $U$, a collection $T =
      \{S_1,\dots,S_n\}$ of subsets of~$U$, $k\in \mathbb{N}$.
    \item[Question:]
      Is there a cover $X \subseteq T$ of~$U$ of size $|X| \le k$? 
  \end{description}
\end{problem}

In the following, we will first prove several general lemmas
concerning properties of optimal \abbrv{DAG} compressions; we believe
these lemmas to be interesting in their own right as they tell us more
about the power and limitations of \abbrv{DAG} compressions. The proof
of Theorem~\ref{theorem:main} is presented afterwards. At the end of the
section we present a variant of the theorem in a dynamic setting, which
shows that it is not only hard to compute an optimal \abbrv{DAG}
compression ``from scratch'', but even updating an already existing
one even without changing the cluster \abbrv{DAG} is hard.

\subparagraph*{Properties of Optimal \abbrv{DAG} Compressions.}

The lemmas proved in the following all state that for all graphs~$\ori
G$, possibly satisfying some restrictions, there \emph{exists} an
optimal \abbrv{DAG} compression with certain properties. The proofs
always start by considering an optimal \abbrv{DAG} compression that
violates the claimed properties and then argues that we can ``fix''
the violation by changing the compression slightly to arrive at a new
\abbrv{DAG} compression that is still optimal. By possibly repeating
the process, we get an optimal \abbrv{DAG} compression that satisfies
the property.
For the first property recall that a \emph{twin} in a graph $G = (V,E)$ is a pair $t_1$,~$t_2$
of vertices such that their in- and out-neighborhoods are
identical; see the left part of Figure~\ref{fig:recons} for an example.
\begin{lemma~}[Twins Will Be Twins]\label{lemma:TwinsWillBeTwins}
  Every graph $\ori G$ has an optimal \abbrv{DAG} compression
  $(V,A,E)$ such that all twins in $\ori G$ are also twins in
  $(V,A)$ and in $(V,E)$.
\end{lemma~}
\begin{proof~}
  Suppose there are twins $\{\ori{t}_1,\ori{t}_2\}$ in~$\ori{G}$
  that are not twins in $(V,A)$ or in $(V,E)$. Without loss of
  generality, assume that $\ori{t}_1$ has  the smaller total degree
  (the sum of the sizes of its 
  in-neighborhood and out-neighborhood in $(V,A)$ and in $(V,E)$) of
  the two vertices. Then remove all arcs and compression edges
  incident to $\ori{t}_2$ and add the same arcs and compression edges incident
  to $\ori{t}_1$ also to $\ori{t}_2$. Since we picked the twin with the smallest
  total degree, we deleted at least as many arcs and compression edges as we
  added. Therefore, the resulting \abbrv{DAG} compression remains optimal, and it
  is easy to see that it still represents the same graph~$\ori{G}$. We can
  iterate this process until all twins in~$\ori{G}$ are also twins
  in $(V,A)$ and in~$(V,E)$.
\end{proof~}

\begin{figure}[htpb]
  \centering
  
  \newcommand{\twinsRAW}{

    \node[node] (1) {$\ori{1}$};
    \node[node, right of=1] (2) {$\ori{2}$};
    \node[node, right of=2] (3) {$\ori{3}$};

    \node[node, above=3cm of 1] (1tS1) {$\ori{t}_{1}$};
    \node[node, right=0.2cm of 1tS1] (2tS1) {$\ori{t}_{2}$};

    \node[clusterVertex, above left=0.5cm and 0.2cm of 3] (c4) {};

    \node[clusterVertex, below =1.3cm of 1tS1] (c5) {};
    
    \node[node, grayed,  right= 0.5cm of 2tS1] (1tS2) {};
    \node[node, grayed,  right=0.2cm of 1tS2] (2tS2)  {};

    \node[clusterVertex, grayed,  below right =0.6cm and -0.05cm of 1tS2] (tc2) {};

    \node[clusterVertex,  below right =0.6cm and 0.4cm of tc2 ] (tc4) {};

    \draw (1tS2) edge[compressed, grayed] (tc2);
    \draw (2tS2) edge[compressed, grayed] (tc2);
    \draw (tc4) edge[compressed, grayed] (c4);

    \draw[arc] (c4) -- (2);
    \draw[arc] (c4) -- (3);

    
    \draw (c5) edge[cluster] (1);
    
    \ifStageOne
    \draw (1tS1) edge[compressed] (c4);
    \draw (1tS1) edge[compressed]  (c5);
    \draw (2tS1) edge[compressed] (2)
                 edge[compressed, bend left=20] (3);
    \draw (2tS1) edge[compressed]  (c5);
    \fi
    
    \ifStageTwo
    \draw (2tS1) edge[compressed] (c4);
    \draw (1tS1) edge[compressed] (c4);
    \draw (1tS1) edge[compressed]  (c5);
    \draw (2tS1) edge[compressed]  (c5);
    \fi

    \ifStageThree
    \node[clusterVertex, above right =0.4 cm and 0cm of c5 ] (tc) {$c$};
    \draw (1tS1)   edge[compressed] (tc);
    \draw (2tS1)   edge[compressed] (tc);
    \draw (tc) edge[cluster] (c5);
    \draw (tc) edge[cluster] (c4);
    \fi

    
    \path[path picture={
        \fill[white, path fading=fade right] (path picture bounding box.south west) rectangle (path picture bounding box.north east);
    }] (current bounding box.south -| 2.4,0) rectangle (current bounding box.north -| 3.6,0);
  }

  \newif\ifStageOne
  \StageOnetrue 
  \StageOnefalse

  \newif\ifStageTwo
  \StageTwotrue 
  \StageTwofalse

  \newif\ifStageThree
  \StageThreetrue 

  \begin{tikzpicture}[node distance=1.3cm]

    \StageOnetrue
    \StageTwofalse
    \StageThreefalse
    \begin{scope}
      \twinsRAW
    \end{scope}
    \draw[-To] (3.2,2) -- (4,2) node[midway, above=1mm] {Lemma  \ref{lemma:TwinsWillBeTwins}};

    \StageOnefalse
    \StageTwotrue
    \StageThreefalse
    \begin{scope}[xshift=4.9cm]
      \twinsRAW
    \end{scope}
    \draw[-To] (8.1,2) -- (8.9,2) node[midway, above=1mm] {Lemma \ref{lemma:oneedge}};

    \StageOnefalse
    \StageTwofalse
    \StageThreetrue
    \begin{scope}[xshift=9.8cm]
      \twinsRAW
    \end{scope}
  \end{tikzpicture}
  \caption{Transformation steps that
    maintain an optimal \abbrv{DAG} compression $(V,A,E)$. In the
    first step we make the twins $\ori{t}_1$ and $\ori{t}_2$ in $(\ori
    V, \ori E)$ also twins $(V,A)$ and $(V,E)$ by mirroring the twin
    with the smaller total degree to the twin with the larger total
    degree. In the second step, we add a new cluster vertex~$c$ and
    reconnect edges to lower the number of compression edges incident
    to the twins.}
  \label{fig:recons}

\end{figure}
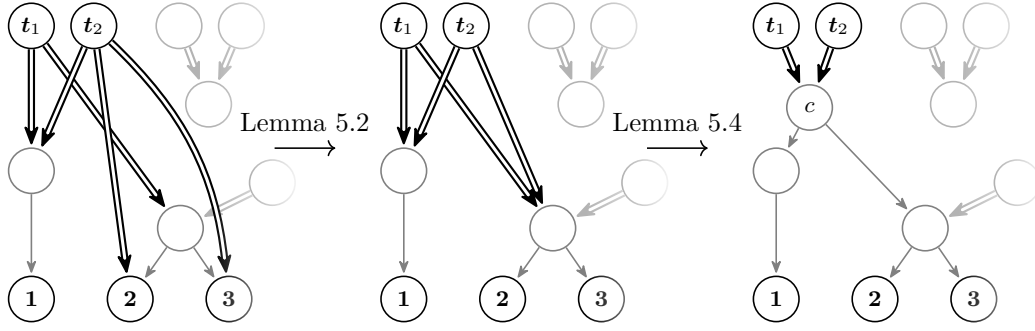

The second property concerns the cluster \abbrv{DAG} of compressions of
``directed bipartite graphs''. Such graphs have a vertex set that can
be partitioned into two \emph{shores $\ori V_1$ and $\ori V_2$} and
the edges are only from $\ori V_1$ to $\ori V_2$, that is, $\ori E
\subseteq \ori V_1 \times \ori V_2$. See Figure~\ref{example2:shores}
for an example.
\begin{lemma~}[Clusters in One Shore Only]\label{lemma:oneShore}
  Every directed bipartite graph $\ori G = (\ori V_1 \cup \ori
  V_2,\ori E)$ has an optimal \abbrv{DAG} compression $(V,A,E)$ such
  that for all $v \in V \setminus \ori{V}_1$ we have $\ori C(v)
  \subseteq  \ori V_2$.
\end{lemma~}
\begin{proof~}
  Let $U_1 = \{ u \in V \mid \ori C(u) \cap \ori{V}_1 \neq
  \emptyset\}$ be the cluster vertices that contain a vertex
  from~$\ori V_1$ and define $U_2$ analogously for $\ori V_2$. Note
  that $U_1$ and $U_2$ must be disjoint in any optimal \abbrv{DAG}
  compression since no compression edge can use a cluster vertex~$v$
  that is present in both, meaning that we could safely remove $v$ and
  reduce the size of the compression.

  The claim now simply states that there is an optimal \abbrv{DAG}
  compression with $U_1 = \ori V_1$. Suppose 
  that this not yet the case. Consider any $u_1 \in U_1 \setminus
  \ori V_1$ that is a source in the cluster \abbrv{DAG} restricted
  to~$U_1$ (such as $u_1 = a$ or $u_1 = d$ in
  Figure~\ref{example2:shores}). Then all compression edges in~$E$
  involving $u_1$ are of the form $(u_1,v) \in E$ with $v \in U_2$ and all cluster arcs
  in~$A$ involving $u_1$ are of the form $(u_1,v)$ with $v \in
  U_1$. Now we \emph{switch this,} meaning that we form $E'$ from~$E$
  and $A'$ from~$A$ by removing all compression edges and all
  cluster arcs involving $u_1$, and instead add the following: For
  each former $(u_1,v) \in E$ we add $(u_1,v)$ to~$A'$, and for each
  former $(u_1,v) \in A$ we add $(v,u_1)$ to~$E'$ (note the
  directions).

  We claim that the transformation yields a new optimal \abbrv{DAG}
  compression of the original graph. Clearly, the size does not
  change. To see that the same edges~$\ori E$ are still represented, consider
  the former compression edges $(u_1,v) \in E$. Jointly, they
  represented the edges (recall that  $\ori C(u_1) =
  \bigcup_{(u_1,u) \in A} \ori C(u)$):  
  \begin{align*}
    \bigcup_{(u_1,v) \in E} \bigl(\ori C(u_1) \times \ori
    C(v)\bigr) =
    \left(\bigcup_{(u_1,u) \in A} \ori C(u)\right) \times
    \left(\bigcup_{(u_1,v) \in E} \ori C(v)\right).
  \end{align*}
  In the new compression, the new compression edges jointly
  represent (where $\ori C'(\cdot)$ is the cluster function in the new 
  graph)
  \begin{align*}
    \bigcup_{(v,u_1) \in E'} \bigl(\ori C'(v) \times \ori
    C'(u_1)\bigr) =
    \left(\bigcup_{(v,u_1) \in E'} \ori C'(v)\right) \times
    \left(\bigcup_{(u_1,u) \in A'} \ori C'(u)\right).
  \end{align*}
  Comparing the two last values and noting that $\ori C(x) = \ori
  C'(x)$ for all $x \neq u_1$, we see that they are the same.

  By repeating the process as often as needed, we get a compression
  with $U_1 = \ori V_1$.

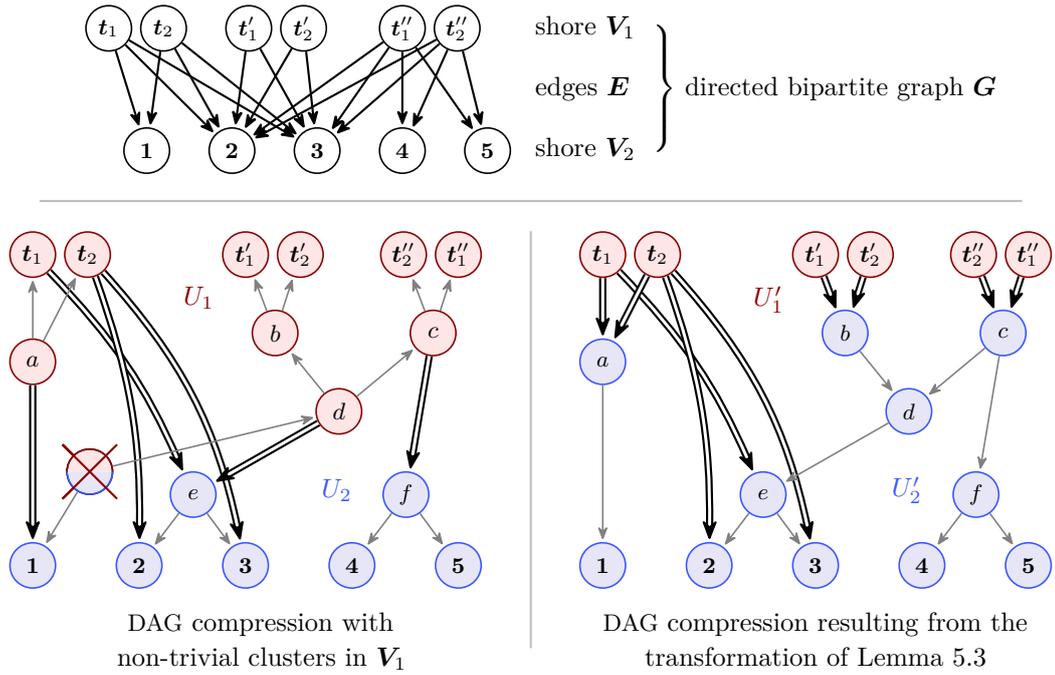
\begin{figure}[htpb]
  \centering

  \newcommand{\DAGCompressionTikz}{
    
    \node[elementClusterVertex] (1) {$\ori{1}$};
    \node[elementClusterVertex, right of=1] (2) {$\ori{2}$};
    \node[elementClusterVertex, right of=2] (3) {$\ori{3}$};
    \node[elementClusterVertex, right of=3] (4) {$\ori{4}$};
    \node[elementClusterVertex, right of=4] (5) {$\ori{5}$};
    
    \node[twinCluster, above=3.5cm of 1] (1tS1) {$\ori{t}_{1}$};
    \node[twinCluster, right=0.1cm of 1tS1] (2tS1) {$\ori{t}_{2}$};
    
    \node[twinCluster,  above=3.5cm of 3] (1tS2) {$\ori{t}_{1}'$};
    \node[twinCluster,  right=0.1cm of 1tS2] (2tS2)  {$\ori{t}_{2}'$};
    
    \node[twinCluster, above=3.5cm of 5] (2tS3)  {$\ori{t}_{1}''$};
    \node[twinCluster,  left=0.1cm of 2tS3] (1tS3) {$\ori{t}_{2}''$};

    \node[elementClusterVertex, above left=0.5cm and 0.25cm of 3] (c4) {$e$};
    \node[elementClusterVertex, above left=0.5cm and 0.25cm of 5] (c5) {$f$};
    
    \ifStageOne
    \node[twinCluster, below =0.8cm of 1tS1] (tc1) {$a$};
    \node[twinCluster, below right =0.6cm and -0.05cm of 1tS2] (tc2) {$b$};
    \node[twinCluster, below right =0.6cm and -0.05cm of 1tS3] (tc3) {$c$};
    \node[twinCluster, below right =0.6cm and 0.4cm of tc2] (tc4) {$d$};
    
    \node[text=darkred, below left=1mm of 1tS2] {$U_1$};
    \node[text=darkblue] at (4,1) {$U_2$};
    \fi
    
    \ifStageTwo
    \node[elementClusterVertex, below =0.8cm of 1tS1] (tc1) {$a$};
    \node[elementClusterVertex, below right =0.6cm and -0.05cm of 1tS2] (tc2) {$b$};
    \node[elementClusterVertex, below right =0.6cm and -0.05cm of 1tS3] (tc3) {$c$};
    \node[elementClusterVertex,  below right =0.6cm and 0.4cm of tc2 ] (tc4) {$d$};
    
    \node[text= darkred, below left=1mm of 1tS2]  {$U_1'$};
    \node[text= darkblue] at (4,1) {$U_2'$};
    
    \fi
    
    \draw  (c4) edge[cluster] (2);
    \draw  (c4) edge[cluster] (3);
    \draw  (c5) edge[cluster] (4);
    \draw  (c5) edge[cluster] (5);
    
    \draw(1tS1) edge[compressed, bend left =10] ( c4);
    \draw(2tS1) edge[compressed,  bend left =15] (3);
    \draw(2tS1) edge[compressed,  bend left =10] (2);

    \ifStageOne

    \node[node, below left =1.4cm and 2cm of tc2, path picture={
        \begin{scope}
          \clip (path picture bounding box.center) -- (path picture bounding box.east) arc (0:180:1cm) -- cycle;
          \fill[lightred] (path picture bounding box.south west) rectangle (path picture bounding box.north east);
        \end{scope}
        \begin{scope}
          \clip (path picture bounding box.center) -- (path picture bounding box.west) arc (180:360:1cm) -- cycle;
          \fill[lightblue] (path picture bounding box.south west) rectangle (path picture bounding box.north east);
        \end{scope}
    }] (sup) {};
    \draw[darkred, semithick] (sup.center) ++(0:0.3cm) arc (0:180:0.3cm);
    \draw[darkblue, semithick] (sup.center) ++(180:0.3cm) arc (180:360:0.3cm);

    \doublecross{sup}
    
    \draw (tc1) edge[arc] (1tS1);
    \draw(tc1) edge[arc] (2tS1);
    \draw(tc2) edge[arc] (1tS2);
    \draw(tc2) edge[arc] (2tS2);
    \draw(tc3) edge[arc] (1tS3);
    \draw(tc3) edge[arc] (2tS3);

    \draw(tc4) edge[arc] (tc2);
    \draw(tc4) edge[arc] (tc3);
    
    \draw(sup) edge[arc] (tc4);
    \draw(sup) edge[arc] (1);
    
    \draw(tc1) edge[compressed](1);
    \draw(tc4) edge[compressed] (c4);
    \draw(tc3) edge[compressed] (c5);

    \fi

    \ifStageTwo
    
    \draw (1tS1) edge[compressed] (tc1);
    \draw (2tS1) edge[compressed](tc1);
    \draw(1tS2) edge[compressed](tc2);
    \draw (2tS2) edge[compressed](tc2);
    \draw (1tS3) edge[compressed](tc3);
    \draw(2tS3) edge[compressed](tc3);

    \draw (tc1) edge[cluster] (1);
    \draw (tc4) edge[cluster] (c4);
    \draw(tc3) edge[cluster] (c5);
    
    \draw (tc2) edge[cluster] (tc4);
    \draw (tc3) edge[cluster] (tc4);
    
    \fi
  }
  
  \newif\ifStageOne
  \StageOnetrue 
  \StageOnefalse
  
  \newif\ifStageTwo
  \StageTwotrue 
  \StageTwofalse

  \begin{tikzpicture}[node distance=1.4cm,,
      elementClusterVertex/.style={node,draw= darkblue,fill=lightblue},
      twinCluster/.style={ node,draw= darkred, fill=lightred, font=\small}]
    \centering

    \begin{scope}[xshift=1.5cm, yshift=5.5cm]
      \node[vertex] (1) {$\ori{1}$};
      \node[vertex, right=5mm of 1] (2) {$\ori{2}$};
      \node[vertex, right=5mm of 2] (3) {$\ori{3}$};
      \node[vertex, right=5mm of 3] (4) {$\ori{4}$};
      \node[vertex, right=5mm of 4] (5) {$\ori{5}$};

      
      \node[vertex, above=1cm of 1,xshift=-5mm] (1tS1) {$\ori{t}_{1}$};
      \node[vertex, right=1mm of 1tS1] (2tS1) {$\ori{t}_{2}$};
      
      \node[vertex, right=5mm of 2tS1, ] (1tS2) {$\ori{t}_{1}'$};
      \node[vertex, right=1mm of 1tS2] (2tS2)  {$\ori{t}_{2}'$};
      
      \node[vertex, above=1cm of 4] (2tS3)  {$\ori{t}_{1}''$};
      \node[vertex, right=1mm of 2tS3] (1tS3) {$\ori{t}_{2}''$};
      
      \draw[dedge] (1tS1) -- (1);
      \draw[dedge] (1tS1) -- (2);
      \draw[dedge] (1tS1) -- (3);
      \draw[dedge] (2tS1) -- (1);
      \draw[dedge] (2tS1) -- (2);
      \draw[dedge] (2tS1) -- (3);

      \draw[dedge] (1tS2) -- (2);
      \draw[dedge] (1tS2) -- (3);
      \draw[dedge] (2tS2) -- (2);
      \draw[dedge] (2tS2) -- (3);

      \draw[dedge] (1tS3) -- (2);
      \draw[dedge] (1tS3) -- (3);
      \draw[dedge] (1tS3) -- (4);
      \draw[dedge] (1tS3) -- (5);
      \draw[dedge] (2tS3) -- (2);
      \draw[dedge] (2tS3) -- (3);
      \draw[dedge] (2tS3) -- (4);
      \draw[dedge] (2tS3) -- (5);
    \end{scope}
    \node[at={(5|-2tS3)},anchor=mid west,xshift=5mm] {shore $\ori{V}_1$};
    \node[above=5mm of 5,anchor=mid west,xshift=5mm] {edges $\ori E$};
    \node[above=5mm of 5,anchor=mid west,xshift=2cm,overlay] {$\left.\vrule height10mm width0pt\right\}$ directed bipartite graph $\ori{G}$};
    \node[at=(5),anchor=mid west,xshift=5mm] {shore $\ori{V}_2$};

    \StageOnetrue
    \StageTwofalse
    \begin{scope}
      \DAGCompressionTikz
    \end{scope}
    \node[anchor=north,text width=6cm,align=flush center] at (3, -0.5) {\abbrv{DAG} compression with
      non-trivial~clusters in~$\ori V_1$};
    
    \StageOnefalse
    \StageTwotrue
    \begin{scope}[xshift=7.5cm]
      \DAGCompressionTikz
    \end{scope}
    \node[anchor=north,text width=6cm,align=flush center] at (10.3, -0.5) {
      \abbrv{DAG} compression resulting from the transformation of Lemma~\ref{lemma:oneShore}};
    
    \draw[sepline] ([yshift=-2.6cm,xshift=4mm]current bounding box.north west) -- ([yshift=-2.6cm,xshift=-4mm]current bounding box.north east);
    \draw[sepline] ([yshift=-3cm]current bounding box.north) -- ([yshift=4mm]current bounding box.south);
    
  \end{tikzpicture}

  \caption{Example of the transformation from
    Lemma~\ref{lemma:oneShore} to ensure that an optimal \abbrv{DAG}
    compression of a directed graph $\ori G = (\ori V_1 \cup \ori
    V_2,\ori E)$ only has ``clusters in shore $\ori V_2$'', meaning
    that $\ori C(v) \subseteq \ori V_2$ holds for all $v \in V
    \setminus \ori V_1$. Recall, that $U_1$ and $U_2$ are defined
    as the set of cluster vertices whose clusters intersect $\ori{V}_1$ and
    $\ori{V}_2$, respectively. Note that the vertex in $U_1 \cap U_2$
    has no incident compression edge and can be removed. An example of
    the ``switch'' from Lemma~\ref{lemma:oneShore} happens for~$a$: The
    compression edge $(a,\ori 1) \in E$ is replaced by an arc $(a,\ori
    1) \in A'$, while the two arcs $(a,\ori t_1) \in A$ and $(a,\ori
    t_2) \in A$ get replaced by compression edges $(\ori t_1,a) \in E'$
    and $(\ori t_2,a) \in E'$ in the opposite direction. Other switches
    are for, first,~$d$, followed by $b$ and~$c$,
    resulting in the compression shown right.    
  }
  \label{example2:shores}

\end{figure}
\end{proof~}
Note that the \abbrv{DAG} compression from Lemma~\ref{lemma:oneShore}
has the property that \emph{all compression edges go from a vertex
in~$\ori V_1$ to a cluster vertex in $V \setminus \ori V_1$.}

The third lemma is an interesting combination of the first two lemmas,
see the right side of Figure~\ref{fig:recons} for an example:
\begin{lemma~}[Twins Have Only One Compression Edge]\label{lemma:oneedge}
  Every directed bipartite graph $\ori G = (\ori V_1 \cup \ori 
  V_2,\ori E)$ has an optimal \abbrv{DAG} compression $(V,A,E)$ such
  that all twins $\ori t_1,\ori t_2 \in \ori V_1$ have at most one
  incident compression edge.
\end{lemma~}
\begin{proof~}
  Let an optimal \abbrv{DAG} compression $(V,A,E)$ of~$\ori G$ be
  given. By Lemma~\ref{lemma:TwinsWillBeTwins} we may assume that all
  twins of~$\ori G$ are also twins in $(V,A)$ and in $(V,E)$. By
  Lemma~\ref{lemma:oneShore} we may assume that all $v \in V \setminus
  \ori V_1$ have $\ori C(v) \subseteq \ori V_2$; and note that the
  transformation of the second lemma does not destroy the ``twins will
  be twins'' property. Now suppose that there are twins $\ori t_1,\ori
  t_2 \in \ori V_1$ such that $(\ori t_1, u) \in \ori E$ and also
  $(\ori t_1, v) \in \ori E$ for some $u \neq v$. Since $\ori t_1$ and
  $\ori t_2$ are also twins in $(V,A)$, we also have $(\ori t_2,
  u),(\ori t_2, v) \in \ori E$. Remove these four compression edges,
  add a new cluster vertex $c$ to~$V$, add two cluster arcs $(c,u)$ and
  $(c,v)$ to~$A$, and add two compression edges $(\ori t_1,c)$ and
  $(\ori t_2,c)$ to~$E$. It is easy to see that we now still have an
  optimal \abbrv{DAG} compression of~$\ori G$, but the degree of
  $\ori t_1$ and $\ori t_2$ in~$E$ has been reduced by one. If we
  repeat the transformation as long as possible, we get a graph
  satisfying the claim.
\end{proof~}


\subparagraph*{The Reduction.}
We are now nearly ready to present the reduction from
$\Lang{set-cover}$ to $\Lang{min-dag-compression}$ and prove
Theorem~\ref{theorem:main}. One final definition and lemma will be needed: The standard way of
encoding a collection~$T$ of subsets of~$U$ as a graph is through
the bipartite \emph{incidence graph,} where one shore is~$U$ and the
other has a vertex $a_S$ for each $S \in T$ and there are edges from
each~$a_S$ to all elements of~$S$. For our purposes, 
it will be useful to have a ``twinned'' version of the
incidence graph, where $a_S$ has an additional twin~$b_S$:
\begin{definition}
  Let $T$ be a collection of subsets of~$U$. The \emph{twinned
  incidence graph $\ori T = (\ori V_1 \cup \ori V_2,\ori E)$ of~$T$}
  is the directed bipartite graph with $\ori V_1 =
  \{\ori a_S \mid S \in T\} \cup \{\ori b_S \mid S \in   T\}$,  $\ori V_2 = U$, and $\ori E = \{(\ori
  a_S,\ori s) \mid \ori s \in S \in T\} \cup  \{(\ori 
  b_S,\ori s) \mid \ori s \in S \in T\}$.
\end{definition}

\begin{lemma~}\label{lemma-twinned-size} Let $T$ be a collection of subsets
  of~$U$. Let $T'= T \cup \{R\}$ for a set $R \subseteq U$ that is not a subset
  of any $S \in T$. Let $k_=$ be the minimum size of a set $X_= \subseteq T$
  such that $\bigcup X_= = R$ (and let $k_= = \infty$ if no such set exists) and
  let $k_{\supseteq}$ be the minimum size of a set  $X_\supseteq \subseteq T$
  such that $\bigcup X_\supseteq \supseteq R$. Let $s$ and $s'$ be the sizes the
  optimal \abbrv{DAG} compressions of the twinned incidence graphs of~$T$
  and~$T'$, respectively. Then $s + k_{\supseteq} + 2 \le s'\le s + k_= + 2$. 
\end{lemma~}
\begin{proof~}
  Consider an optimal \abbrv{DAG} compression $(V,A,E)$ of the
  twinned incidence graph $\ori T = (\ori V_1 \cup \ori V_2,\ori
  E)$ of~$T$. By Lemma~\ref{lemma:oneedge} we may assume that the 
  twins of~$\ori T$, meaning in particular \emph{all} vertices in~$\ori V_1$,
  have only one incident compression edge in~$E$. This implies that each
  compression edge $(\ori a_S, c_S) \in E$ must yield all of $\ori a_S
  \times S \subseteq \ori E$ and, hence, $\ori C(c_S) = S$. In other
  words, \emph{for each $S \in T$ there is a cluster vertex $c_S \in
  V$ with $\ori C(c_S) = S$.}

  To see that $s' \le s + k_= + 2$ holds, form $(V',A',E')$ from 
  $(V,A,E)$ as follows: Let $V' = V \cup \{\ori a_R,\ori b_R,c_R\}$,
  let $E'=E \cup (\{\ori a_R,\ori b_R\} \times \{c_R\})$, and let $A'=A
  \cup \{(c_R,c_S) \mid S \in X_=\}$. Clearly, $|A'| + |E'| = |A| +
  k_= + |E| + 2$ and $(V',A',E')$ is a \abbrv{DAG} compression of the
  twinned incidence graph of~$T'$.

  To see that $s' \ge s + k_{\supseteq} + 2$, start with an optimal
  \abbrv{DAG} compression $(V',A',E')$ of the twinned incidence graph of~$T'$ and,
  as before, let it satisfy the property that for each $S \in T'$
  there is a cluster vertex $c_S \in V'$ with $\ori C'(c_S) =
  S$. In particular, there must be a vertex~$c_R$ with $\ori C'(c_R) = 
  R$ and compression edges $(\ori a_R,c_R)$ and~$(\ori
  b_R,c_R)$. Consider the arcs $(c_R,v) \in A'$: We must have $\ori
  C'(v) \subseteq S \in T$ since, otherwise, we could get a better
  compression by directly connecting $c_R$ to $v$'s children and
  removing~$v$. In particular, the number of $v$ with $(c_r,v) \in A'$
  is at least $|X_{\supseteq}|$ as the children of $c_R$ must
  cover~$R$, so each $\ori C'(c_R)$ is contained in some $S \in T$,
  implying that there are at least $|X_{\supseteq}|$ many such arcs
  (possibly more).  
  This means that we can build a \abbrv{DAG} compression $(V,A,E)$
  of~$T$'s twinned incidence graph of size at most $s' - k_{\supseteq} - 2$ by 
  removing $\ori a_R$ and $\ori b_R$ along with their two incident
  compression edges, and removing all arcs $(c_R,v) \in A'$.
\end{proof~}

\begin{proof}[Proof of Theorem~\ref{theorem:main}]
  We only prove the $\Class{NP}$-hardness of
  $\Lang{min-dag-compression}$ by reducing $\Lang{set-cover}$ to
  it. Let $T$ be an input collection of subsets of~$U$ and let $k$ be
  a number. We may assume that some trivial cases are taken care of,
  namely that $\bigcup T = U$ (that is, $U$ can be
  covered at all), but also $U \notin T$. We also assume that $U =
  \{1,\dots,n\}$ holds.  

  The crucial idea is to impose some structure on~$T$ by replacing it
  by $\{\{s\} \mid s \in T\} \cup \bigcup_{i=1}^n \{S \cap
  \{1,\dots,i\} \mid S \in T\} \setminus \{\emptyset\}$. Note that we essentially add
  to $T$ all singleton sets and all ``non-empty initial segments of all $S \in T$'' where the
  initial segment is with respect to ordering of 
  the universe~$U$. For instance, $T = \{\{1,4,5,6\}, \{2,3, 5,
  7\}\}$, would be replaced by
  $\bigl\{\{1\},\dots,\{7\},\penalty0 \{1,4\},\penalty0 \{2,3\},\penalty0 \{1,4,5\},\penalty0 \{2,3,5\},\penalty0 \{1,4,5,6\},\penalty0  \{2,3,5,7\}\bigr\}$. Note
  that these added sets are not helpful with respect to covering~$U$
  with as few sets as possible in~$T$, so the original $T$ lies in
  $\Lang{set-cover}$ iff the new one does. For the resulting
  $T$, let us number the elements as $T = \{S_1,\dots, S_m\}$, where the
  $S_i$ are in standard order; that is, the first sets $S_1$ to $S_n$ have size~$1$,
  followed by all $S_i$ of size~$2$, followed by all of size~$3$, and
  so on.  

  Let $T_i := \{S_1,\dots, S_i\}$ be the collection encompassing only
  the first $i$ sets, let $\ori T_i$ be the twinned incidence graph
  of~$T_i$, and let $s_i$ be the size of an optimal \abbrv{DAG}
  compression of~$\ori T_i$. Clearly, $s_n = 2n$ (for each $\{j\} \in T_n$
  both $\ori a_{\{j\}}$ and $\ori b_{\{j\}}$ can and must be connected by a
  compression edge to $j \in U$). By Lemma~\ref{lemma-twinned-size} we
  have $s_{i-1} + X_{\supseteq} + 2 \le s_i \le s_{i-1} + X_= + 2$
  for all $i\in\{n+1,\dots, m\}$. However, for these~$i$ we have
  $|X_{\supseteq}| = |X_m| = 2$ as any $S_i$ can be covered exactly by two
  sets in $T_{i-1}$, namely the initial segment of $S_i$ missing the 
  highest-number $j \in S_i$ and the singleton $S_j = \{j\}$, but
  cannot be covered by any single set in $T_{i-1}$. All
  told, $s_i = s_{i-1} + 4$ for $i\in\{n+1,\dots,
  m\}$, and thus $s_m = 2n + 4(m-n-1) = 4m -2n -4$.

  The reduction now asks whether the twinned incidence graph of $T
  \cup \{U\}$ has an optimal \abbrv{DAG} compression of size $4m - 2n
  -2 + k = s_m + k + 2$. We can apply Lemma~\ref{lemma-twinned-size}
  and note that since $R = U$ is the whole universe, any $X_\supseteq$
  is also an $X_=$. In particular, such an optimal \abbrv{DAG}
  compression exists iff a set cover $X\subseteq T$ of $U$ exists of
  size $|X| \le k$.
\end{proof}

\subparagraph*{Hardness of Updating Optimal Compressions.}
Recall from the introduction that it is of independent interest to
show that updating an existing optimal \abbrv{DAG} compression to
reflect the addition or deletion of a single edge to the original graph is also
$\Class{NP}$-complete:

\begin{theorem~}\label{theorem:update-add}
  $\Lang{min-dag-compression-add}$ is $\Class{NP}$-complete.
\end{theorem~}

\begin{proof~}
  We start with the same idea as in the proof of
  Theorem~\ref{theorem:main}, where we reduced $\Lang{set-cover}$ to
  $\Lang{min-dag-compression}$. Just as in that proof, we start with a
  collection $T$ of subsets of~$U$, but now let $U = \{2,\dots,
  n+1\}$. Replace $T$ by $T' = \{\{1\} \cup S \mid S \in T\}$, that is, we
  add $1$ to all sets in our  collection and call such sets
  \emph{infected.} Clearly, this will not change the size of an
  optimal set cover, but now of $\{1\} \cup U$. For this new set~$T'$,
  we proceed as in the proof of Theorem~\ref{theorem:main} and add
  sets to~$T'$ so that $\hat{T}' = \{S_1,\dots, S_m\}$ for sets $S_i$ of
  increasing size and such that all initial segments of any $S_i$ are
  already in~$\hat{T}'$. Crucially, observe that all 
  sets in
  $\hat{T}'$ are infected, that is, contain~$1$. Finally, recall from the
  proof of Theorem~\ref{theorem:main} that the twinned incident graph
  $\ori{\hat{T}'} = (\ori V_1 \cup \ori V_2, \ori E)$ of~$\hat{T}'$ has an optimal 
  \abbrv{DAG} compression $(V,A,E)$ of size $4m -2(n+1) -4$ where $n+1
  = \left|U\cup\{1\}\right|$ and observe that we can easily
  compute this optimal \abbrv{DAG} compression. 

  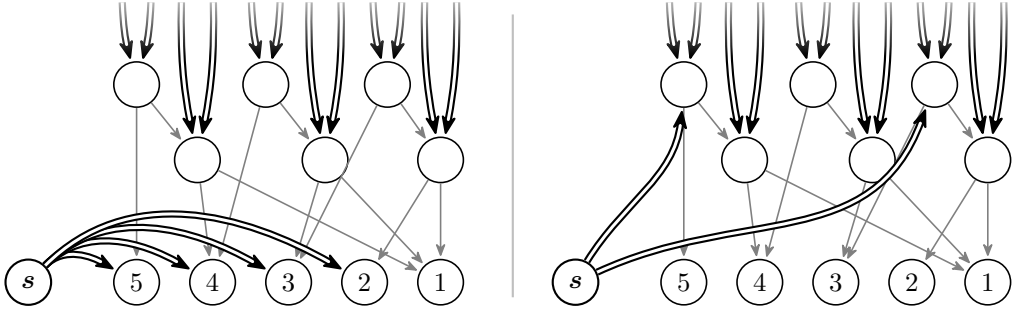
\begin{figure}[htpb]
  \centering
  \tikzfading[name=fade down, top color=transparent!0, bottom color=transparent!100]

  \newcommand{\AddFixed}{
    \centering

    \tikzset{ twinCluster/.style={sinkVertex, node, gray},
      clusterVertex/.style = {node}}


    \ifStageOne
    \tikzset{infected/.style={}}
    \fi

    \ifStageTwo
    \tikzset{infected/.style={}}
    \fi

    \clip (-1.8,-.6) rectangle (4.5,3.95);

    \node[sinkVertex,] (1) {$5$};
    \node[sinkVertex, right of=1] (2) {$4$};
    \node[sinkVertex, right of=2] (3) {$3$};
    \node[sinkVertex, right of=3] (4) {$2$};
    \node[sinkVertex,infected, right of=4] (x) {$1$};
    \node[sinkVertex,thick, left= 0.8cm of 1] (s) {$\ori{s}$};

    \def\sheight{2cm}
    \node[clusterVertex,infected, above  =\sheight and 0cm of 1] (S1) {};
    \node[clusterVertex,infected, above =\sheight and 0cm of 2, xshift =0.7cm] (S2) {};
    \node[clusterVertex,infected, above  =\sheight and 0cm of 4, xshift =0.3cm] (S3) {};

    \def\sbelowX{0.8cm}
    \def\sbelowY{1cm}
    \node[clusterVertex,infected, above = 1cm of 2, xshift=-0.2cm ] (S1b) {};
    \node[clusterVertex,infected, above = 1cm of 3 , xshift=0.48cm ] (S2b) {};
    \node[clusterVertex,infected,above = 1cm of x] (S3b) {};

    \ifStageOne
    \fi

    \draw (S1b) edge[arc] (x);
    \draw (S2b) edge[arc] (x);
    \draw (S3b) edge[arc] (x);
    \draw (S1b) edge[arc] (2);
    \draw (S2b) edge[arc] (3);
    \draw (S3b) edge[arc] (4);
    \draw (S1) edge[arc] (S1b);
    \draw (S2) edge[arc] (S2b);
    \draw (S3) edge[arc] (S3b);
    \draw (S1) edge[arc] (1);
    \draw (S2) edge[arc] (2);
    \draw (S3) edge[arc] (3);

    \tikzset{ compressedG/.style={compressed},
      compressedFade/.style={compressed},
      twinCluster/.style={}
    }
    \def\verticalDist{1.5}
    \def\halfDist{0.2}

    \node  (1tS1)  at ([yshift=\verticalDist cm, xshift=-\halfDist cm] S1)  [twinCluster] { };
    \node (2tS1)    at ([yshift=\verticalDist cm, xshift=\halfDist cm] S1) [twinCluster] {};
    
    \node  (1tS2)  at ([yshift=\verticalDist cm, xshift=-\halfDist cm] S2)  [twinCluster] { };
    \node (2tS2)    at ([yshift=\verticalDist cm, xshift=\halfDist cm] S2) [twinCluster] {};

    \node  (1tS3)  at ([yshift=\verticalDist cm, xshift=-\halfDist cm] S3)  [twinCluster] { };
    \node (2tS3)    at ([yshift=\verticalDist cm, xshift=\halfDist cm] S3) [twinCluster] {};

    \def\verticalDist{2.5}
    \def\halfDist{0.2}
    
    \node  (1tS1b)  at ([yshift=\verticalDist cm, xshift=-\halfDist cm] S1b)  [twinCluster] {};
    \node (2tS1b)    at ([yshift=\verticalDist cm, xshift=\halfDist cm] S1b) [twinCluster] {};
    
    \node  (1tS2b)  at ([yshift=\verticalDist cm, xshift=-\halfDist cm] S2b)  [twinCluster] {};
    \node (2tS2b)    at ([yshift=\verticalDist cm, xshift=\halfDist cm] S2b) [twinCluster] {};

    \node  (1tS3b)  at ([yshift=\verticalDist cm, xshift=-\halfDist cm] S3b)  [twinCluster] {};
    \node (2tS3b)    at ([yshift=\verticalDist cm, xshift=\halfDist cm] S3b) [twinCluster] {};

    \def\bend{6}
    \draw (1tS1) edge[compressedFade, bend right= \bend] (S1);
    \draw (2tS1) edge[compressedFade, bend left = \bend] (S1);
    \draw (1tS2) edge[compressedFade, bend right= \bend] (S2);
    \draw (2tS2) edge[compressedFade, bend left = \bend] (S2);
    \draw (1tS3) edge[compressedFade, bend right= \bend] (S3);
    \draw (2tS3) edge[compressedFade, bend left = \bend] (S3);
    \draw (1tS1b) edge[compressedFade, bend right= \bend] (S1b);
    \draw (2tS1b) edge[compressedFade, bend left = \bend] (S1b);
    \draw (1tS2b) edge[compressedFade, bend right= \bend] (S2b);
    \draw (2tS2b) edge[compressedFade, bend left = \bend] (S2b);
    \draw (1tS3b) edge[compressedFade, bend right= \bend] (S3b);
    \draw (2tS3b) edge[compressedFade, bend left = \bend] (S3b);

    \path[path picture={
        \fill[white, path fading=fade down] (path picture bounding box.south west) rectangle (path picture bounding box.north east);
    }] (-0.5,3) rectangle (4.5,4);

    \ifStageOne
    \draw (s) edge[compressedG, out=45, in= 150,  looseness=1] (1);
    \draw (s) edge[compressedG,  out=45, in= 150,  looseness=1] (2);
    \draw (s) edge[compressedG,  out=45, in= 150,  looseness=1] (3);
    \draw (s) edge[compressedG, out=45, in= 150,  looseness=1] (4);
    \fi

    \ifStageTwo
    \draw (s) edge[compressed, bend right = 25, out =0] (S1);
    \draw (s) edge[compressed, in = 250, out =25, looseness=1.1] (S3);
    \fi
    
  }

  \newif\ifStageOne
  \StageOnetrue 
  \StageOnefalse

  \newif\ifStageTwo
  \StageTwotrue 

  \begin{tikzpicture}[node distance=1cm, minimum size=0.8cm]
    \StageOnetrue
    \StageTwofalse
    \begin{scope}[xshift=-4cm, xshift=1cm, yshift = 1.6cm]
      \AddFixed

    \end{scope}

    \StageOnefalse
    \StageTwotrue
    \begin{scope}[xshift=4.2cm,  yshift = 1.6cm]
      \AddFixed
    \end{scope}
    
    \draw[sepline] ([yshift=-4mm]current bounding box.north) -- ([yshift=4mm]current bounding box.south);
    
  \end{tikzpicture}

  \caption{
    Construction from Theorem~\ref{theorem:update-add}: Left, we have an
    optimal \abbrv{DAG} compression of $\hat{T}'$'s twinned incidence graph,
    where all sets in $T$ contain the special element~$1$, joined with
    edges from a special vertex $\ori s$ to all elements in
    $\{2,\dots,n+1\}$ (but, \emph{not,} to~$1$). The fact that there
    is no edge to $1$ from~$\ori s$ means that no ``infected'' cluster
    vertex~$c$ (meaning $\ori s \in \ori C(c)$) may be used in the
    compression, implying that the shown compression is
    optimal. Adding the single edge $(\ori s,1)$ to $\ori E$ changes
    the situation dramatically: Now we can compress $\{\ori s\} \times
    \{1,\dots, n+1\}$ by connecting $\ori s$ to the cluster vertices
    of a minimal set cover.
  }
  \label{FixedAddFig}

\end{figure}

  At this point, we diverge from the previous proof: We build a graph
  $\ori G = (\ori V_1'\cup\ori V_2, \ori E')$ from $\ori{\hat{T}'}$ by adding
  a new vertex $\ori s$ to $\ori V_1$, so $\ori V_1'= \ori V_1 \cup
  \{\ori s\}$, and adding edges $(\ori s,u)$ to $\ori E$ for all $u
  \in \{2,\dots,n+1\} = U$, so $\ori E'= \ori E \cup \{(\ori s,u) \mid
  u\in U\}$ (but  note that $(\ori s,1) \notin \ori E'$). Clearly $D=(V
  \cup \{\ori s\},A, E \cup \{(\ori s,u) \mid u\in U\})$ is a \abbrv{DAG}
  compression of~$\ori G$ of size $4m -2(n+1) -4 + n$.

  We claim that the compression is optimal. To see this, let any
  optimal \abbrv{DAG} compression $(V_0,A_0,E_0)$ of~$\ori G$ be
  given. By Lemmas~\ref{lemma:TwinsWillBeTwins}, \ref{lemma:oneShore},
  and \ref{lemma:oneedge}, for every $S \in \hat{T}'$ with $|S| \ge 2$
  there is a cluster vertex $c_S \in V_0$ with $\ori C(c_S) =
  S$. There can be no compression edge $(\ori s,c_S)$ for any of these
  $c_S$ since all these $S$ are infected; meaning that we can only
  have $(\ori s,c) \in \ori E_0$ for other cluster vertices~$c$. In
  particular, if we remove all of these $m - (n+1)$ many $c_S$ together with their
  incident arcs and compression edges, we must still have the
  compression edges and arcs representing $\{\ori s\} \times
  U$. Each $c_S$ has two incoming compression edges and at least two
  outgoing arcs, meaning that we remove $4m-4n-4$ edges, leaving $2n +
  n$ edges. If we, next, remove all $\ori a_S$ and $\ori b_S$ for the
  singleton $S = \{j\}$ with $j \in \{2,\dots,n+1\}$ together with
  their compression edges, we are left with $n$ edges. However, to
  represent $\{\ori s\} \times U$, we need at least $|U| = n$ many
  compression edges and\,/\,or arcs, which shows that our compression
  was optimal.  

  To conclude the proof, we output $\ori G$, the optimal compression
  $D$ of~$\ori G$, the new edge $(\ori s,1)$, and the number $k + 4m
  -2(n+1) -4$. The crucial observation is that the new edge  $(\ori
  s,1)$ means that we now must compress $\{\ori s\} \times (\{1\} \cup
  U)$. It is now easy to see (using the same arguments as before) that
  the best way to do this is simply to connect $\ori s$ via
  compression edges to the clusters of a minimal size set cover $X
  \subseteq T$ of $U \cup \{1\}$. 
\end{proof~}

\begin{theorem~}\label{theorem:update-del}
    $\Lang{min-dag-compression-del}$ is $\Class{NP}$-complete.
\end{theorem~}

\newcommand{\twin}{\ori{t}_{\ori{1}}^{\ori{U} \cup \{\ori{x}\}}}
\begin{proof~} 
    We once more reduce from the \Lang{set-cover} problem. Let $(T,k)$ be an instance of
    \Lang{set-cover} where $T$ is a collection of subsets of $U = \{2, \dots, n
    + 1\}$. First, set $T' = T \cup \{ U \cup \{1\} \}$ and again add sets
    to $T'$ such that $\TSnice = \{S_1, \dots, S_m\}$ consists of sets $S_i$ of
    increasing size and all inititial segments of any $S_i$ are already in
    $\TSnice$. Recall that a set $S \in \TSnice$ is called \emph{infected} if $1
    \in S$.

    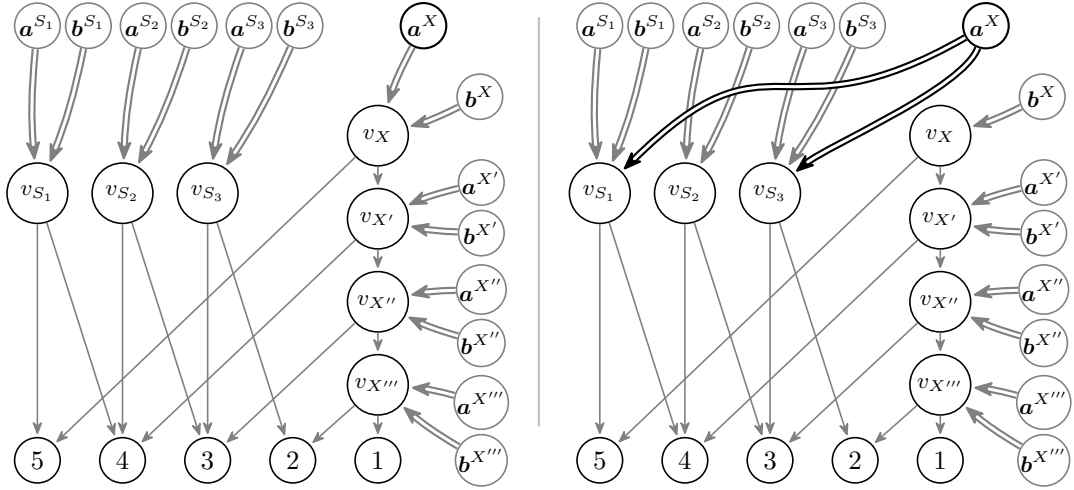
\begin{figure}[htpb]
	\centering
    \newcommand{\FixedDelete}{

    \tikzset{ twinCluster/.style={sinkVertex, inner sep = 0, font=\footnotesize,
    draw=gray, fill=white},
    clusterVertex/.style = {circle, draw, inner sep=0, font=\footnotesize}}

    \node[sinkVertex] (1) {$5$};
    \node[sinkVertex, right of=1] (2) {$4$};
    \node[sinkVertex, right of=2] (3) {$3$};
    \node[sinkVertex, right of=3] (4) {$2$};
    \node[sinkVertex, right of=4] (x) {$1$};

    \def\sheight{2.82cm}
    \node[clusterVertex, above  =\sheight and 0cm of 1] (S1) {$v_{S_1}$};
    \node[clusterVertex, above  =\sheight and 0cm of 2] (S2) {$v_{S_2}$};
    \node[clusterVertex, above  =\sheight and 0cm of 3] (S3) {$v_{S_3}$};

     \def\xspacing{0.28cm}
    \node[clusterVertex, above = \xspacing of x] (x4) {$v_{X'''}$};
    \node[clusterVertex, inner sep=0, above = \xspacing of x4] (x3) {$v_{X''}$};
    \node[clusterVertex, inner sep=0, above = \xspacing of x3] (x2) {$v_{X'}$};
    \node[clusterVertex, inner sep=0, above = \xspacing of x2] (x1) {$v_{X}$};

    \draw (x4) edge[arc] (x);
    \draw (x4) edge[arc,  bend left = 0] (4);
    \draw (x3) edge[arc] (x4);
    \draw (x3) edge[arc,  bend left = 0] (3);
    \draw (x2) edge[arc] (x3);
    \draw (x2) edge[arc,  bend left = 0] (2);
    \draw (x1) edge[arc] (x2);
    \draw (x1) edge[arc,  bend left = 0] (1);
    
    \draw (S1) edge[arc] (1);
    \draw (S1) edge[arc] (2);
    \draw (S2) edge[arc] (2);
    \draw (S2) edge[arc] (3);
    \draw (S3) edge[arc] (3);
    \draw (S3) edge[arc] (4);

        \tikzset{ compressedG/.style={compressed,gray}}
        \def\close{0.05cm}
        \def\step{0.1cm}
        \node[twinCluster, above =1.5cm and 0cm of S1] (1tS1) { $\ori{a}^{S_1}$};
        \node[twinCluster, right=\close of 1tS1] (2tS1) {$\ori{b}^{S_1}$};
        
        \node[twinCluster, right=\step of 2tS1] (1tS2) {$\ori{a}^{S_2}$};
        \node[twinCluster, right=\close of 1tS2] (2tS2)  {$\ori{b}^{S_2}$};
    
        \node[twinCluster, right=\step of 2tS2] (1tS3) {$\ori{a}^{S_3}$};
        \node[twinCluster, right=\close of 1tS3] (2tS3)  {$\ori{b}^{S_3}$};

        \node[twinCluster, right = 0.7cm of x] (1tX4) { $\ori{b}^{X'''}$};
        \node[twinCluster, above=\close of 1tX4] (2tX4) {$\ori{a}^{X'''}$};
        
        \node[twinCluster, above=\step of 2tX4] (1tX3) {$\ori{b}^{X''}$};
        \node[twinCluster, above=\close of 1tX3] (2tX3)  {$\ori{a}^{X''}$};
    
        \node[twinCluster, above=\step  of 2tX3] (1tX2) {$\ori{b}^{X'}$};
        \node[twinCluster, above=\close of 1tX2] (2tX2)  {$\ori{a}^{X'}$};

        \node[twinCluster, draw=black, fill=white, thick, right = 1 of 2tS3] (1tX1)  {$\ori{a}^{X}$};
        \node[twinCluster, below right=0.5 and 0.3 of 1tX1] (2tX1) {$\ori{b}^{X}$};

        \draw (1tS1) edge[compressedG, bend right = 5] (S1);
        \draw (2tS1) edge[compressedG,  bend left = 5] (S1);
            \draw (1tS2) edge[compressedG, bend right = 5  ] (S2);
        \draw (2tS2) edge[compressedG, bend left = 5] (S2);
            \draw (1tS3) edge[compressedG,  bend right = 5 ] (S3);
        \draw (2tS3) edge[compressedG, bend left = 5] (S3);

        \draw (1tX2) edge[compressedG, bend left = 5] (x2);
        \draw (2tX2) edge[compressedG, bend right = 5] (x2);
        \draw (1tX3) edge[compressedG, bend left = 5] (x3);
        \draw (2tX3) edge[compressedG, bend right = 5] (x3);
          \draw (1tX4) edge[compressedG, bend left = 5] (x4);
        \draw (2tX4) edge[compressedG, bend right = 5] (x4);

        \draw (2tX1) edge[compressedG, bend left = 5] (x1);
        \ifStageOne
        \draw (1tX1) edge[compressedG, bend right = 5] (x1);
        \fi
    
        \ifStageTwo
        \draw (1tX1) edge[compressed, in = 45, out = 210,  looseness = 1.5] (S1);
        \draw (1tX1) edge[compressed, in = 40, out = 250,  looseness = 0.5] (S3);
        \fi

    }

    \newif\ifStageOne
    \StageOnetrue 
    \StageOnefalse
    
    \newif\ifStageTwo
    \StageTwotrue 

    \begin{tikzpicture}[node distance=1.12cm, minimum size=0.8cm]
    \centering
        
        \StageOnetrue
        \StageTwofalse
        \begin{scope}
            \FixedDelete
        \end{scope}
    
        \StageOnefalse
        \StageTwotrue
        \begin{scope}[xshift=7.4cm]
            \FixedDelete
        \end{scope}

        \draw[sepline] ([yshift=-2mm]current bounding box.north) -- ([yshift=8mm]current bounding box.south);

    \end{tikzpicture}

   \caption{Construction from Theorem~\ref{theorem:update-del}: Left, we have an
    optimal \abbrv{DAG} compression of $\hat{T}'$'s twinned incidence graph.
    By
    deleting the edge $(\ori{a}^{X}, 1)$, we may no longer use the infected
    vertices in the optimal \abbrv{DAG} compression, which means that the new
    optimal compression is given by connecting $\ori{a}^{X}$ to the minimal set
    cover.}
    \label{FixedDeleteFig}
\end{figure}

    Proceeding, we construct the twinned incidence graph $\ori{\hat{T}'}$ of $\TSnice$ like in
    the proof of Theorem~\ref{theorem:main}, with a minimal \abbrv{DAG} compression
    $D=(V,A,E)$ of size $4m -2(n+1) -4$.
    
    We  choose~$\ori{e}=(\ori a_{U \cup \{1\}}, 1)$ as the to be
    removed edge. Moreover, we set $k'= k + 4m -2(n+1) -4$, and the
    reduction output for $\Lang{min-dag-compression-del}$ is
    $(D,\ori{e},k')$.

    We show the correctness of this reduction. First, let $(T,k)$ be an instance
    of the \Lang{set-cover} problem. Therefore, there is a set cover $T^*$ of
    size at most $k$, covering $U$. Thus, the \abbrv{DAG} compression
    $D'=(V,A,E')$ with $E'$ being obtained from $E$ by removing the compressed
    edge $( \ori a_{U \cup \{1\}}, v_{{U} \cup \{1\}})$ and adding the compressed
    edges  $\{\ori a_{U \cup \{1\}}\} \times \{ v_S \mid \ori{C}(v_S)\in T^*\} $ is a \abbrv{DAG}
    compression of $\ori{\TSnice}$ of $(\ori{V},\ori{E} \setminus \{
    \ori{e} \})$, with $|E'| \leq 2\cdot|T|+k$.

    Second, let $(D,\ori{e},k')$ be an instance of the
    $\Lang{min-dag-compression-del}$ problem and $\ori{\TSnice}$ be the
    graph that $D$ represents. Then, let $D'=(V,A,E')$ be a \abbrv{DAG}
    compression of $\ori{\TSnice}$ with $\ori{e}$ removed that has at most
    $k'$ compressed edges. Note that $\ori a_{U \cup \{1\}}$ can have at most $k$ outgoing
    compressed edges in $E'$, each of the $2\cdot|T|-1$ other 
    twins must have at least one outgoing compressed edge in $E'$. Moreover,
    $\ori a_{U \cup \{1\}}$ cannot have a compressed edge in $E'$ to an infected vertex, since
    this would represent edge $\ori{e}$. Therefore, for every vertex $v$ with
    $(\ori a_{U \cup \{1\}},v)\in E'$, there is a set $S \in T$, with $C(v) \subseteq S$ and
    the collection of every of theses at most $k$ sets~$S \in T$ is therefore a
    set cover of $U$.
\end{proof~}

\section{Conclusion}

In this paper, we investigated the complexity of fundamental algorithms on
\abbrv{DAG} compressions. We showed that computing a \abbrv{DAG} compression of
minimal size is $\Class{NP}$-complete, and that even the task of updating a
\abbrv{DAG} compression to deal with an added or deleted edge to the underlying
\abbrv{DAG} is $\Class{NP}$-hard. Indeed, a closer look at our proof of
Theorems~\ref{theorem:update-add} and~\ref{theorem:update-del} shows that even
if we additionally insist that during an update the cluster \abbrv{DAG} remains
the same, the problem is still $\Class{NP}$-complete. On the positive side, we
showed how to construct minimum spanning trees in almost linear time
with respect to the size of the \abbrv{DAG} compression rather than
the size of the original graph. We also saw that 
\abbrv{DAG} compressions yield strictly smaller compression results than tree
compressions.

The general intractability of finding minimal \abbrv{DAG} compressions
leads to further interesting questions: In the light of the result of
Bannach et al.~\cite{BannachMT24} that minimal \abbrv{DAG}
compressions may be easily constructed on graphs of bounded
twin-width when the corresponding contraction sequence is given, one
can ask which further structural properties such as information about
cliques, bicliques, and clusters with a common neighborhood in the
graph could be utilized to construct minimal \abbrv{DAG} compressions.

We showed in our paper that \abbrv{DAG} compressions are strictly better
than tree compressions. To better understand the difference in power of these two
compression types, a characterization of the graph families on which the
minimal \abbrv{DAG} compression is strictly smaller than the minimal tree
compression would be an interesting further research direction.

Another interesting direction would be to further extend the
algorithmic toolbox that uses \abbrv{DAG} compressions, thereby
providing even more faster fundamental algorithms for graphs. A
specific example would be the computation of a maximum flow on a
\abbrv{DAG} compression.

\bibliography{main}

\tcsautomoveinsert{main}

\section{Supplementary Algorithms and Figures for Kruskal's
Algorithm}\label{section:kruskal-appendix}

A \emph{union-find data structure $P$} keeps track of
a partition $\{\ori U_1,\dots, \ori U_k\}$ of~$\ori V$; for $\ori v
\in \ori V$ let us write $P(\ori v)$ for the set $\ori U_i \subseteq
\ori V$ containing~$\ori v$. In slight abuse of notation, we also
write $\ori U \in P$ to denote that $\ori U$ is one of the sets in
the partition. The data structure is initialized with the partition into singleton sets,
that is, with $P(\ori v) = \{\ori v\}$ for all $\ori v\in \ori V$. Two operations are
supported: First, $P.\Algo{find}(\ori v)$ returns some internal value with 
the guarantee that $P.\Algo{find}(\ori u) = P.\Algo{find}(\ori v)$ iff $P(\ori u) =
P(\ori v)$. Second, $P.\Algo{unite}(\ori u_1,\ori u_2)$ changes the partition to $P'$
so that $P'(\ori u_1) = P'(\ori u_2) = P(\ori u_1) \cup P(\ori u_2)$ and $P'(\ori v) = P(\ori v)$
for $\ori v \notin P(\ori u_1) \cup P(\ori u_2)$; in other words, $P(\ori u_1)$ and
$P(\ori u_2)$ are united in $P'$. It is well-known~\cite[Chapter 21.4]{CLRS} that a
union-find data structure can be implemented in such a way that a
sequence of $k$ many $\Algo{find}$ and $\Algo{unite}$ operations
takes time $k \cdot \alpha(|\ori V|)$, where $\alpha$ is the inverse Ackermann
function. 

\begin{figure}[htpb]
  \centering
  \tikzset{
    clean/.style={fill=white},
    unknown/.style={fill=black!25},
    clean a/.style=unknown,
    clean b/.style=unknown,
    clean c/.style=unknown,
  }
  \newcommand{\MSTExampleDAGC}{
    
    \node[sinkVertex] (1) {$\ori{1}$};
    \node[sinkVertex, right of=1] (2) {$\ori{2}$};
    \node[sinkVertex, right of=2] (3) {$\ori{3}$};
    \node[sinkVertex, right =0.2cm of 3] (4) {$\ori{4}$};
    \node[sinkVertex, right of=4] (5) {$\ori{5}$};
    \node[sinkVertex, right of=5] (6) {$\ori{6}$};
    \node[sinkVertex, right =0.2cm of 6] (7) {$\ori{7}$};
    
    \node[clusterVertex, clean a, above =0.7cm of 2] (c1) {$a$};
    \node[clusterVertex, clean b, above =0.5cm of 5] (c2) {$b$};
    \node[clusterVertex, clean c, above left = 1.9cm and 0cm of 7] (c3) {$c$};

    \draw[cluster] (c1) -- (1);
    \draw[cluster] (c1) -- (2);
    \draw[cluster] (c1) -- (3);
    \draw[cluster] (c2) -- (4);
    \draw[cluster] (c2) -- (5);
    \draw[cluster] (c2) -- (6);
    \draw[cluster] (c3) -- (7);
    \draw[cluster] (c3) -- (c2);
    
    \ifStageOne
    \draw (c1)  edge[compressedUn, "1", swap]  (c2);
    \draw (c1)  edge[compressedUn, "2"]    (c3);
    \fi
    \ifStageTwo
    \draw (c1)  edge[red, compressedUn, "1", swap]  (c2);
    \draw (c1)  edge[compressedUn, "2"]  (c3);
    \fi
    \ifStageThree
    \draw (c1)  edge[red, compressedUn, "1", swap]  (c2);
    \draw (c1)  edge[compressedUn, "2"]  (c3);
    \fi
    \ifStageFour
    \draw (c1)  edge[compressedUn, compressedUn, "1", swap]  (c2);
    \draw (c1)  edge[red, compressedUn, "2"] (c3);
    \fi
  }

  \newcommand{\MSTExampleGraph}{

    \begin{scope}[sinkVertex/.append style={fill=white}]
      \node[sinkVertex,yshift=3mm] (1) {$\ori{1}$};
      \node[sinkVertex, below=3mm of 1] (2) {$\ori{2}$};
      
      \node[sinkVertex, below=3mm of 2] (3) {$\ori{3}$};

      \node[sinkVertex, right = 1cm of  1] (4) {$\ori{4}$};
      \node[sinkVertex, below= 3mm of 4] (5) {$\ori{5}$};
      \node[sinkVertex, below= 3mm of 5] (6) {$\ori{6}$};
      
      \node[sinkVertex, left = 0.7cm of  2 ] (7) {$\ori{7}$};
    \end{scope}
    
    \begin{scope}[on background layer]
      \begin{scope}[blue!20,line width=9mm,line join = round, line cap=round]
        \ifStageOne
        \filldraw (1.center) -- (1.center);
        \filldraw (2.center) -- (2.center);
        \filldraw (3.center) -- (3.center);
        \filldraw (4.center) -- (4.center);
        \filldraw (5.center) -- (5.center);
        \filldraw (6.center) -- (6.center);
        \filldraw (7.center) -- (7.center);
        \fi

        \ifStageTwo
        \draw (3.center) -- (1.center) -- (4.center);
        \filldraw (5.center) -- (5.center);
        \filldraw (6.center) -- (6.center);
        \filldraw (7.center) -- (7.center);
        \fi

        \ifStageThree
        \filldraw (1.center) -- (3.center) -- (6.center)
        -- (4.center) -- cycle;
        \filldraw (7.center) -- (7.center);
        \fi

        \ifStageFour
        \filldraw (1.center) -- (4.center) -- (6.center) -- (3.center)
        -- (7.center) -- cycle;
        \fi
      \end{scope}
    \end{scope}
    
    \begin{scope}[dedge,-]
    \ifStageOne
    \draw(1) edge[]  (4);
    \draw[](1) -- (5);
    \draw[](1) -- (6);
    
    \draw[](2) -- (4);
    \draw[](2) -- (5);
    \draw[](2) -- (6);

    \draw[](3) -- (4);
    \draw[](3) -- (5);
    \draw[](3) -- (6);
    
    \draw(1) edge["\footnotesize 2", swap]   (7);
    \draw (2) edge["\footnotesize 2", swap] (7);
    \draw (3) edge["\footnotesize 2"] (7);
    \fi

    \ifStageTwo
    \draw[pickedOne](1) -- (4);
    \draw[](1) -- (5);
    \draw[](1) -- (6);
    
    \draw[pickedOne](2) -- (4);
    \draw[](2) -- (5);
    \draw[](2) -- (6);

    \draw[pickedOne](3) -- (4);
    \draw[](3) -- (5);
    \draw[](3) -- (6);
    
    \draw (1) edge["\footnotesize 2", swap] (7);
    \draw (2) edge["\footnotesize 2", swap] (7);
    \draw (3) edge["\footnotesize 2"] (7);
    \fi

    \ifStageThree
    \draw[pickedOne](1) -- (4);
    \draw[pickedOne](1) -- (5);
    \draw[pickedOne](1) -- (6);
    
    \draw[pickedOne](2) -- (4);
    \draw[](2) -- (5);
    \draw[](2) -- (6);

    \draw[pickedOne](3) -- (4);
    \draw[](3) -- (5);
    \draw[](3) -- (6);
    
    \draw (1) edge["\footnotesize 2", swap] (7);
    \draw (2) edge["\footnotesize 2", swap] (7);
    \draw (3) edge["\footnotesize 2"] (7);
    \fi

    \ifStageFour
    \draw[pickedOne](1) -- (4);
    \draw[pickedOne](1) -- (5);
    \draw[pickedOne](1) -- (6);
    
    \draw[pickedOne](2) -- (4);
    \draw[](2) -- (5);
    \draw[](2) -- (6);

    \draw[pickedOne](3) -- (4);
    \draw[](3) -- (5);
    \draw[](3) -- (6);
    
    \draw (1) edge[pickedOne, "\footnotesize 2", swap] (7);
    \draw (2) edge["\footnotesize 2", swap] (7);
    \draw (3) edge["\footnotesize 2"] (7);
    \fi

    \end{scope}

  }

  \begin{tikzpicture}[node distance=0.9cm]
    \centering

    \newif\ifStageOne
    \newif\ifStageTwo
    \newif\ifStageThree
    \newif\ifStageFour

    \newcommand{\mstblock}[3]{

      \begin{scope}[xshift = 0cm,yshift= 0cm]
        
	\MSTExampleDAGC
	
      \end{scope}
      \begin{scope}[xshift = 9cm, yshift=1.8cm]
	\MSTExampleGraph
	
      \end{scope}
      
      \node[anchor=mid] at (2.7, -0.8) { #1 };
      \node[anchor=mid] at (9.2, -0.8) { #2 };
      \node[anchor=mid] at (5.5, -0.8) { #3 };

    }

    \begin{scope}
      \StageOnetrue
      \mstblock{\abbrv{DAG} compression $D=(V,A,E)$ of $\ori G$}{Graph
        $\ori G = (\ori V, \ori E)$, partition \textcolor{blue!50}{$P$}}{}
      \draw[sepline] (-0.8,-1.2) -- (11.5,-1.2);
    \end{scope}

    \begin{scope}[yshift=-4.5cm]
      \StageTwotrue
      \tikzset{clean a/.style=clean}
      \mstblock{}{}{Process $\textcolor{red}{(a,b)} \in E$: First, clean $a$ by
        connecting each child to $\ori c(b) = \ori 4$}
      \draw[sepline] (-0.8,-1.2) -- (11.5,-1.2);
    \end{scope}

    \begin{scope}[yshift=-9cm]
      \StageThreetrue
      \tikzset{clean a/.style=clean, clean b/.style=clean}
      \mstblock{}{}{Process $\textcolor{red}{(a,b)} \in E$: Second, clean $b$ by
        connecting each child to $\ori c(a) = \ori 1$}
      \draw[sepline] (-0.8,-1.2) -- (11.5,-1.2);
    \end{scope}

    \begin{scope}[yshift=-13.5cm]
      \StageFourtrue
      \tikzset{clean a/.style=clean, clean b/.style=clean, clean c/.style=clean}
      \mstblock{}{}{Process $\textcolor{red}{(a,c)} \in E$: Clean $c$ by
        connecting children $b$ and $\ori 7$ if needed to $\ori c(a) =
        \ori 1$}
    \end{scope}

  \end{tikzpicture}

  \caption{Example computation of a minimum spanning tree (in green) on a \abbrv{DAG}
    compression~$D$ of a graph $\ori G$. Edges with no weight label in $\ori G$
    have weight $1$. The developing partition $P$ is shown in
    blue. Clean vertices in $V$ are shown in white; initially only the
    sinks are white. To process a compression edge like $\{a,b\}$, we
    first make $a$ clean by connecting its child clusters to the
    representative vertex $\ori c(b)$, which is $\ori 4 \in \ori C(b)$ in
    this example. Then we clean $b$. The algorithm would then try to
    add the edge $\{\ori c(a),\ori c(b)\} = \{\ori1,\ori4\}$, but the
    union-find data structure $P$ informs the algorithm that these
    vertices are already in the same set and no edge is
    added. Processing $\{a,b\}$ then means cleaning $c$, which causes no
    edges to be added between $\ori c(a)$ and the cluster of the
    already-clean child $b$, but it causes an edge to be added between
    $\ori c(a)$ and~$\ori 7$. Once more, the final edge $\{\ori
    c(a),\ori c(c)\} = \{\ori1,\ori4\}$ is skipped.
  }
  \label{example:mstComp}

\end{figure}

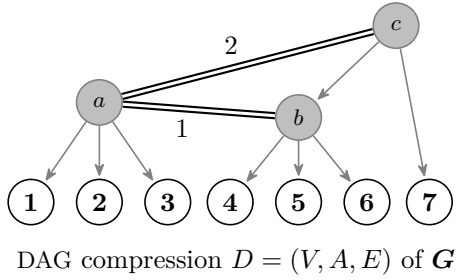
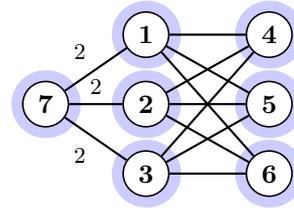
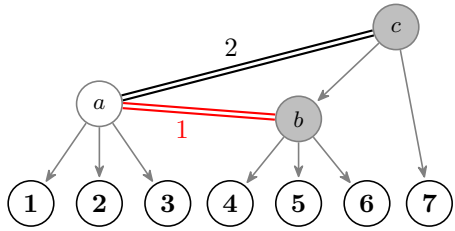
\end{document}